%% file: corr.tex
\documentclass[acmtocl]{acmtrans2m}
\acmYear{15}

\newtheorem{convention}{Convention}[section]
\newtheorem{theorem}{Theorem}[section]

\newtheorem{corollary}[theorem]{Corollary}

\newtheorem{lemma}[theorem]{Lemma}
\newdef{definition}[theorem]{Definition}
\newdef{example}[theorem]{Example}
\newdef{remark}[theorem]{Remark}
\newdef{fact}[theorem]{Fact}
\newdef{claim}[theorem]{Claim}
\newdef{algorithm}[theorem]{Algorithm}
\newdef{property}[theorem]{Property}

\usepackage{amsmath}
\usepackage{amssymb}
\usepackage{subfig}
\usepackage{xspace}
\usepackage{marginnote}

\newcommand\SUCC{\checkmark}
\newcommand\FAIL{\times}
\input{common.inc}

\colorlet{myblue}{blue!25!lightgray}
\colorlet{myred}{red!25!white}

\usepackage{xcolor}

\newcommand\CSI{\textsf{CSI}\xspace}
\newcommand{\seq}[2][n]{{#2_1},\dots,{#2_{#1}}}
\newcommand{\ari}{\mathop{\mathrm{ari}}}
\renewcommand{\ari}{\mathop{\mathrm{arity}}}
\newcommand{\dom}{\mathop{\mathrm{dom}}}
\newcommand{\rt}{\mathop{\mathrm{root}}}
\newcommand{\rank}{\mathop{\mathrm{rank}}}
\newcommand{\RRR}[2][]{\rud{#1}{#2}%
{\mathrel{\blacktriangleright\kern-.2em\blacktriangleright}}}
\newcommand{\LLL}[2][]{\lud{#1}{#2}%
{\mathrel{\blacktriangleleft\kern-.2em\blacktriangleleft}}}
\newcommand{\rrr}[2][]{\rud{#1}{#2}%
{\mathrel{\vartriangleright\kern-.2em\vartriangleright}}}
\renewcommand{\lll}[2][]{\lud{#1}{#2}%
{\mathrel{\vartriangleleft\kern-.2em\vartriangleleft}}}
\newcommand{\RRV}[2][]{\rud{#1}{#2}%
{\mathrel{\blacktriangleright\kern-.2em\vartriangleright}}}
\newcommand{\LLV}[2][]{\lud{#1}{#2}%
{\mathrel{\vartriangleleft\kern-.2em\blacktriangleleft}}}
\newcommand{\x}[1]{\mathcal{#1}}
\newcommand{\m}[1]{\mathsf{#1}}
\renewcommand{\AA}{\x{A}}
\newcommand{\BB}{\x{B}}
\newcommand{\CC}{\x{C}}
\newcommand{\DD}{\x{D}}
\newcommand{\FF}{\x{F}}
\newcommand{\VV}{\x{V}}
\newcommand{\RR}{\x{R}}
\newcommand{\UU}{\x{U}}
\newcommand{\TT}{\x{T}}

\newcommand{\OO}{\x{O}}
\renewcommand{\SS}{\x{S}}
\newcommand{\bL}{\mathbb{L}}
\newcommand{\bLT}[1][]{\bL_{#1} \cap \TT(\FF,\VV)}
\newcommand{\vL}{\bL'}
\newcommand{\vLT}{\vL \cap \TT(\FF,\VV)}
\renewcommand{\vL}{{\mathbb{V}}}
\newcommand{\bLv}{\bL_\vL}
\newcommand{\bLvT}{\bLv \cap \TT(\FF,\VV)}
\newcommand{\Nat}{\mathbb{N}}
\renewcommand{\merge}{\sqcup}

\let\from\leftarrow
\newcommand{\Cleq}{\sqsubseteq}
\newcommand{\Clt}{\sqsubset}

\newcommand{\myvec}{\mathbf}
\newcommand{\lst}[1]{\hat{#1}}
\newcommand{\PP}{\m{PP}}
\newcommand{\Cu}{\m{Cu}}
\newcommand{\rnf}[2][]{#2{\downarrow}_{#1}}
\newcommand{\VC}{\VV_\square}
\renewcommand{\leq}{\leqslant}
\renewcommand{\geq}{\geqslant}
\renewcommand{\succeq}{\geqslant}
\renewcommand{\succ}{>}
\renewcommand{\curlyvee}{\vee}
\newcommand{\matches}{%
\def\next##1##2{\mathrel{{\leqslant}\raise.17ex\hbox to 0pt%
{$##1\hss\cdot$}}}%
\mathpalette\next{}}
\newcommand{\II}{\mathrel{\setbox1=\hbox{$\leftrightarrow$}%
\lower.25ex\hbox to\wd1{$\vdash\hss\dashv$}}}

\newcommand{\ctop}{(L1)\xspace}
\renewcommand{\ctop}{(L$_1$)\xspace}
\newcommand{\cvar}{(L2)\xspace}
\renewcommand{\cvar}{(L$_2$)\xspace}
\newcommand{\cvarp}{(L2$'$)\xspace}
\renewcommand{\cvarp}{(L$_2'$)\xspace}
\newcommand{\cpartial}{(L3)\xspace}
\renewcommand{\cpartial}{(L$_3$)\xspace}
\newcommand{\crewrite}{(W)\xspace}
\newcommand{\cconsistent}{(C1)\xspace}
\renewcommand{\cconsistent}{(C$_1$)\xspace}
\newcommand{\cstepfusion}{(C2)\xspace}
\renewcommand{\cstepfusion}{(C$_2$)\xspace}

\newcommand{\compatat}{compatible$^\star$\xspace}

\title{Layer Systems for Proving Confluence}

\author{%
BERTRAM FELGENHAUER and AART MIDDELDORP and HARALD ZANKL \\
University of Innsbruck
\and
VINCENT VAN OOSTROM \\
Utrecht University
}

\begin{abstract}
We introduce layer systems for proving generalizations of the modularity
of confluence for first-order rewrite systems. Layer systems specify how
terms can be divided into layers. We establish structural conditions on
those
systems that imply confluence. Our abstract framework covers known results
like modularity, many-sorted persistence, layer-preservation and currying.
We present a counterexample to an extension of persistence to order-sorted 
rewriting and derive new sufficient conditions for the extension to hold.
All our proofs are constructive.
\end{abstract}




\begin{document}

\makeatletter
\renewcommand\permission{
Permission to make digital or hard copies of part or all of this work for 
personal or classroom use is granted without fee provided that copies are not 
made or distributed for profit or commercial advantage and that copies bear this 
notice and the full citation on the first page. Copyrights for third-party 
components of this work must be honored. For all other uses, contact the 
Authors.\par
\copyright\ 20\@acmYear\ 
the authors.

\url{http://dx.doi.org/10.1145/2710017}
}
\makeatother

\begin{bottomstuff}
The research described in this paper is supported by
FWF (Austrian Science Fund) project P22467.
\end{bottomstuff}

\maketitle


\section{Introduction}
\label{sec-introduction}

We revisit the celebrated modularity result of confluence,
due to Toyama~\cite{T87}. It states that
the union of two confluent rewrite systems is confluent, provided 
the participating rewrite systems do not share function symbols.
This result has been reproved several times, using category
theory~\cite{L96}, ordered completion~\cite{JT08},
and decreasing diagrams~\cite{vO08b}.
While confluence is also modular for rewriting modulo~\cite{JT08,JL12},
the situation is different for higher-order rewriting~\cite{AvOS10}.
In practice,
modularity is of limited use. More useful techniques,
in the sense that rewrite systems can be
decomposed into smaller systems that share function symbols and rules, are
based on type introduction~\cite{AT97},
layer-preservation~\cite{O94},
and commutativity~\cite{R73}.

Type introduction~\cite{Z94} restricts the set of terms
that have to be considered to the well-typed terms according to some
many-sorted type discipline which is compatible with the rewrite system
under consideration. 
A property of (many-sorted) rewrite systems 
which is preserved and reflected under type removal
is called persistent and Aoto and Toyama~\cite{AT97}
showed that confluence is persistent. In~\cite{AT96}
they extended the latter result by considering an order-sorted type
discipline. However, we show that the
conditions imposed in \cite{AT96} are not sufficient for 
confluence.

The proofs in \cite{O94} and \cite{AT96,AT97} are adaptations of the proof
of Toyama's modularity result by \cite{KMTV94}.
A more complicated proof using concepts from \cite{KMTV94}
has been given by Kahrs, who showed in \cite{K95} that
confluence is preserved under currying~\cite{KKSV96}.
In this article we introduce \emph{layer systems} as a common framework
to capture the results of \cite{AT97,K95,O94,T87} and to identify
appropriate conditions to restore the persistence of confluence for
order-sorted rewriting~\cite{AT96}.
Layer systems identify the parts that are available when decomposing
terms. The key proof idea remains the same. We treat each such layer
independently
from the others where possible, and deal with interactions between
layers separately. The main advantage of and motivation for our proof
is that the result becomes reusable; rather than checking every
detail of a complex proof, we have to check a couple of comparatively
simple, structural conditions on layer systems instead.
Such a common framework also facilitates a formalization of these
results in a theorem prover like Isabelle or Coq.

Besides the theoretical results of this paper we stress
practical implications: Due to an implementation of Theorem~\ref{thm-order}
in our confluence tool \CSI~\cite{ZFM11} it supports a decomposition 
result based on ordered sorts, exceeding the criteria available in 
other tools.
A second result of practical importance is preservation and
reflection of confluence under currying~\cite{K95}, which is used as a
preprocessing step when deciding confluence of ground TRSs~\cite{F12}.

The remainder of this paper is organized as follows. In the next section
we recall preliminaries.
Section~\ref{sec-layer-def} introduces layer systems and establishes 
results how rewriting interacts with layers. 
The main (abstract) results for confluence via layer systems are presented
in Section~\ref{sec-main} and instantiated in
Section~\ref{sec-applications}
to obtain various known results. The new result on order-sorted
persistence is covered in Section~\ref{app-order-sorted2}.
Differences to related work are discussed in Section~\ref{sec-related},
which might be consulted in advance by readers familiar with the
literature. We conclude in Section~\ref{sec-conclusion}.

This article is an extended and significantly revised version of
\cite{FZM11}.
Since here we build upon~\cite{vO08b}, all our proofs are constructive.
Furthermore this 
work is based on a more intuitive definition of layer systems.
The result for non-duplicating systems has been generalized to the strictly
larger class of \emph{bounded duplicating} systems. The application of
quasi-ground systems (Section~\ref{app-quasi-ground}) is new.
Moreover,
all important concepts are demonstrated by examples, and
detailed proofs are provided.

\section{Preliminaries}

We assume familiarity with rewriting~\cite{BN98,TeReSe} and the decreasing
diagrams technique~\cite{vO94}.

Let $\VV$ be a countably infinite set of variables and $\FF$
a signature, i.e., a set of function symbols $f \in \FF$, each associated
with a fixed arity, denoted by $\ari(f)$.
The set of terms over $\FF$ and $\VV$ is denoted by $\TT(\FF,\VV)$. 
The sets of variables and function symbols occurring in a term $t$ are
referred to by $\Var(t)$ and $\Fun(t)$, respectively. 
A term is ground if it does not contain variables. The set of ground
terms over $\FF$ is denoted by $\TT(\FF)$. A term is linear if every
variable occurs at most once.

Let $\square \notin \FF \cup \VV$
be a constant (i.e., a function symbol of arity $0$) called hole
and abbreviate 
$\TT(\FF \cup \{ \square \},\VV)$ by $\CC(\FF,\VV)$.
We write $\VC$ for the set of symbols $\VV \cup \{ \square \}$.
Contexts are terms from $\CC(\FF, \VV)$ containing an arbitrary 
number of holes.
They are partially ordered by $\Cleq$, defined
as the smallest reflexive and transitive relation that is
monotone and satisfies
$\square \Cleq C$ for all $C \in \CC(\FF,\VV)$.
There is a corresponding partial supremum operation,
$\merge$\,, which merges contexts.
The strict order $\Clt$ is defined by $C \Clt D$ if $C \Cleq D$ and
$C \neq D$.
The minimum context with respect to $\Cleq$ is the empty context
$\square$.
By $C[\seq{t}]$ we denote the result of replacing holes in~$C$ by the
terms $\seq{t}$ from left to right. 

The size of a term $t$ is denoted by $|t|$, and
$|t|_W$ for a subset $W \subseteq \FF \cup \VV_\square$
denotes the number of occurrences of function symbols and
variables from $W$ in $t$. We write $|t|_w$ for $|t|_{\{ w \}}$.
Positions of a term $t$ are strings of positive natural numbers,
$\epsilon$ for the
root, and $ip$ if $t = f(t_1,\dots,t_i,\dots,t_n)$ and $p$ is a position
of $t_i$. Then $\Pos(t)$ is the set of all positions of $t$.
Two positions $p$, $q$ are parallel if neither $p$ is a prefix
of $q$ nor $q$ is a prefix of $p$.
Given terms $t$ and $s$,
$t|_p$ is the subterm at position $p$ of $t$
and $t[s]_p$ denotes the result
of replacing $t|_p$ by $s$ in $t$. This operation is extended to sets of
pairwise parallel
positions $P$, resulting in the notation
$t[s_p]_{p \in P}$
By $\rt(t)$ we denote the root symbol of $t$.
For $W \subseteq \FF \cup \VC$ and
$w \in \FF \cup \VC$ we let
$\Pos_W(t) = \{ p \in \Pos(t) \mid \rt(t|_p) \in W \}$ and
$\Pos_w(t) = \Pos_{\{ w \}}(t)$.
A substitution is a map $\sigma\colon \VV \to \TT(\FF,\VV)$
which extends homomorphically to terms.
For terms~$s$ and~$t$ we write $s \matches t$ if there exists a 
substitution~$\sigma$ such that $s\sigma = t$.

A rewrite rule is a pair of terms $(\ell,r) \in \TT(\FF,\VV)^2$,
written $\ell \to r$ such that $\ell \notin \VV$ and
$\Var(\ell) \supseteq \Var(r)$.
A rewrite rule $\ell \to r$ is left-linear if $\ell$ is linear, 
duplicating if there is a variable $x \in \VV$ with $|\ell|_x < |r|_x$, 
and collapsing if its right-hand side is a variable.
A term rewrite system (TRS) consists of a signature and a set of rewrite
rules. If we do not specify differently a TRS will always be over the
signature $\FF$ and variables $\VV$.
The rewrite relation induced by a TRS $\RR$ is denoted $\to_\RR$.
We write $s \to_{p,\ell \to r} t$ if $s \to_\RR t$ using a rule
$\ell \to r \in \RR$ at position $p$.
Two rewrite steps $s \to_\RR t$, $s' \to_\RR t'$ \emph{mirror} each other
if both steps use the same rule at the same position.
This notion is extended to rewrite sequences.
We write $\from$, $\to^=$, $\to^+$ and $\to^*$ to 
denote the inverse, the reflexive closure, the transitive closure, and
the reflexive and transitive closure of a relation~$\to$, respectively.
A relation $\to$ is terminating if $\to^+$ is well-founded and confluent if
${\lud{}{*}\from} \cdot {\to^*} \subseteq {\to^*} \cdot {\lud{}{*}\from}$.
We say that $\to$ is \emph{confluent on} a set $S$ of terms 
if $S$ is closed
under $\to$ and ${\to} \cap ({S \times S})$ is confluent.
A TRS $\RR$ inherits these properties from $\to_\RR$.
A relative TRS~$\RR/\SS$ is a pair of TRSs~$\RR$ and~$\SS$ with the 
induced rewrite relation 
${\to_{\RR/\SS}^{}} =
{\to_\SS^*} \cdot {\to_\RR^{}} \cdot {\to_\SS^*}$. It is
terminating if \smash{$\to_{\RR/\SS}^+$} is well-founded.

Let $\succ$ be a well-founded order on an index set~$I$ and
$\to$ the union of $\to_\alpha$ for all ${\alpha \in I}$.
We write $\rud{\curlyvee\alpha_1 \dots\,\alpha_n}{}\to$ for the union of
$\to_\beta$ where $\alpha_i \succ \beta$ for some
$1 \leqslant i \leqslant n$.
A local peak $t \lud{\alpha}{}\from s \rud{\beta}{}\to u$ is
said to be \emph{decreasing} if 
\[
t \rud{\curlyvee\alpha}{*}\to \cdot \rud{\beta}{=}\to \cdot
\rud{\curlyvee\alpha\beta}{*}\to \cdot
\lud{\curlyvee\alpha\beta}{*}\from \cdot
\lud{\alpha}{=}\from \cdot \lud{\curlyvee\beta}{*}\from u
\]
Furthermore, $\to$ is \emph{locally decreasing} if 
for all $\alpha, \beta \in I$ every local peak
${\lud{\alpha}{}\from} \cdot {\rud{\beta}{}\to}$
is decreasing.
Van Oostrom~\cite{vO94} established the following result.

\begin{theorem}
\label{thm-dd}
Every locally decreasing relation is confluent.
\qed
\end{theorem}

\section{Layer Systems}
\label{sec-layer-def}

In this section we introduce layer systems, which are sets of contexts
satisfying special properties. The top-down decomposition of a term into
maximal layers
admits the notion of the rank of a term. Since for suitable layer systems 
rewriting does not increase the rank this is a valid measure for proofs by
induction.

\begin{definition}
\label{def-top}
Let $\bL \subseteq \CC(\FF,\VV)$ be a set of contexts. Then $L \in \bL$
is called a \emph{top} of a context $C \in \CC(\FF,\VV)$
(according to $\bL$) if $L \Cleq C$.
A top is a \emph{max-top} of $C$ if it is maximal with respect to
$\Cleq$ among the tops of $C$. 
\end{definition}

Note that terms are contexts without holes, so they have tops and
max-tops as well.
In the sequel we use subsets $\bL \subseteq \CC(\FF,\VV)$ to layer terms.
The process is top-down, taking the max-top of a term as layer and
proceeding recursively.

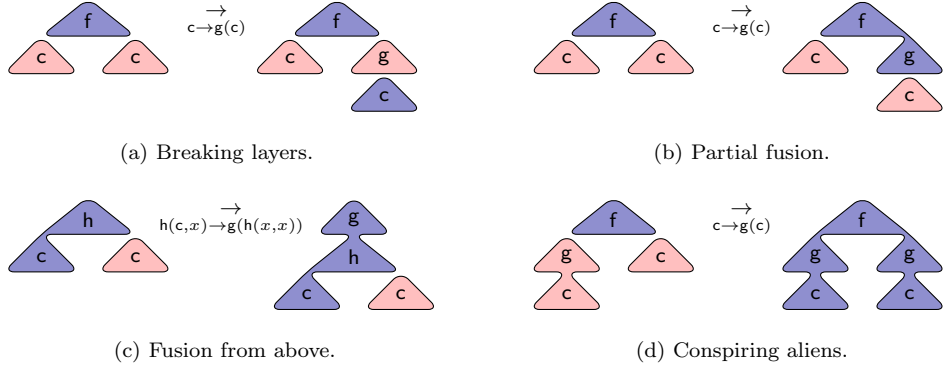
\begin{figure}[t]
\subfloat[\label{fig-problems:bl}Breaking layers.]{%
$
\Vtop{\input{samples/huet_2_1.tikz}}
\underset{\m c \to \m g(\m c)}{\to}
\Vtop{\input{samples/huet_3_1.tikz}}
$}
\hfill
\subfloat[\label{fig-problems:fp}Partial fusion.]{%
$
\Vtop{\input{samples/huet_2_2.tikz}}
\underset{\m c \to \m g(\m c)}{\to}
\Vtop{\input{samples/huet_3_2.tikz}}
$}\\
\subfloat[\label{fig-problems:fa}Fusion from above.]{%
$
\Vtop{\input{samples/fusion_above_l_2.tikz}}
\;\;
\makebox[10mm]{$\underset{\scriptstyle\m h(\m c,x) \to
\m g(\m h(x,x))}{\to}$}
\Vtop{\input{samples/fusion_above_r_2.tikz}}
$}
\hfill
\subfloat[\label{fig-problems:ca}Conspiring aliens.]{%
$
\Vtop{\input{samples/conspiring_l_2.tikz}}
\underset{\m c \to \m g(\m c)}{\to}
\Vtop{\input{samples/conspiring_r_2.tikz}}
$}
\caption{Undesired behavior on layers.}
\label{fig-problems}
\end{figure}

\begin{example}
\label{ex-laysys-new}
Let $\FF$ consist of a binary function symbol $\m f$, a unary
function symbol $\m g$, and constants $\m a$, $\m b$, and $\m c$.
We consider the following candidates for $\bL$:
\begin{align*}
\bL_0 &= \varnothing\\
\bL_1 &= \{ \m f(v,w), \m g(v), \m a, \m b, \m c, v \mid v,w \in \VC \}\\
\bL_2 &= \{ \m f(\m g^n(v),\m g^m(w)), \m g^n(v),\m g^n(\m c),\m a,\m b
\mid v,w \in \VC,\, n,m \in \Nat \}\\
\bL_3 &= \{ \m f(\m g^n(v),\m g^m(w)), \m g^n(v),\m a,\m b
\mid v,w \in \VC \cup \{ \m c \},\, n,m \in \Nat \}
\end{align*}
Regard the terms $s = \m f(\m c,\m c)$ and
$t = \m f(\m c, \m g(\m c))$.
According to $\bL_0$, neither $s$ nor $t$ have any tops.
According to $\bL_1$, the tops of both $s$ and $t$ are $\square$
and $\m f(\square,\square)$, and the latter is the max-top of $s$ and $t$.
According to $\bL_2$, $\square$ and $\m f(\square,\square)$ are tops
of $s$ and $t$, and $\m f(\square,\m g(\square))$ is a top of $t$
but not of $s$. The max-tops of $s$ and $t$ are $\m f(\square,\square)$
and $\m f(\square,\m g(\square))$, respectively. Finally, the max-tops of
$s$ and $t$ according to $\bL_3$ are $s$ and $t$ themselves.
\end{example}

Our goal is to deduce confluence of $\RR$ 
when rewriting is restricted to $\bLT$.
To this end, we need to impose restrictions on $\bL$.
This leads to the central definition of the paper.

\begin{definition}
\label{def-laysys}
Let $\FF$ be a signature.
A set
$\bL \subseteq \CC(\FF,\VV)$ of contexts is called a \emph{layer system}
if it satisfies properties \ctop, \cvar, and \cpartial.
The elements of $\bL$ are called \emph{layers}.
A TRS $\RR$ over $\FF$ is
\emph{weakly layered (according to a layer system~$\bL$)} if 
condition~\crewrite is satisfied for each $\ell \to r \in \RR$.
It is \emph{layered (according to a layer system~$\bL$)} if 
conditions~\crewrite, \cconsistent, and \cstepfusion are satisfied.
The conditions are as follows:
\begin{itemize}
\item[\ctop]
Each term in $\TT(\FF,\VV)$ has a non-empty top.
\item[\cvar]
If $x \in \VV$  and $C \in \CC(\FF,\VV)$ then
$C[x]_p \in \bL$ if and only if $C[\square]_p \in \bL$.
\item[\cpartial]
If $L, N \in \bL$, $p \in \Pos_\FF(L)$, and $L|_p \merge N$ is defined
then $L[L|_p \merge N]_p \in \bL$.
\item[\crewrite]
If $M$ is a max-top of $s$, $p \in \Pos_\FF(M)$, and
$s \to_{p,\ell \to r} t$ then
$M \to_{p,\ell \to r} L$ for some ${L \in \bL}$.
\item[\cconsistent]
In \crewrite either $L$ is a max-top of $t$ or $L = \square$.
\item[\cstepfusion]
If $L, N \in \bL$ and $L \Cleq N$ then $L[N|_p]_p \in \bL$
for any $p \in \Pos_\square(L)$.
\end{itemize}
\end{definition}

\begin{example}[(Example~\ref{ex-laysys-new} revisited)]
Consider the TRS $\RR_1$ consisting of the rewrite rules
\begin{xalignat*}{3}
\m{f}(x,x) &\to \m{a} &
\m{f}(x,\m{g}(x)) &\to \m{b} &
\m{c} \to \m{g}(\m{c})
\end{xalignat*}
from \cite{H80}.
It is non-confluent because
$\m a \lud{\RR_1}{}\from \m f(\m c,\m c) \to_{\RR_1}
\m f(\m c,\m g(\m c)) \to_{\RR_1} \m b$,
and $\m a$, $\m b$ are in normal form.
However, $\RR_1$ is confluent 
on $\bLT[0]$ and $\bLT[2]$,
and $\RR_1$ is confluent if rewriting is restricted to terms of $\bLT[1]$
(which rules out the 
rewrite
step $\m c \to \m g(\m c)$)
but $\RR_1$ is not confluent on $\bLT[3]$,
because $\m a,\m f(\m c,\m c),\m f(\m c,\m g(\m c)),\m b \in \bLT[3]$.
Clearly $\bL_0$ violates \ctop,
and therefore any attempt of proving confluence of $\RR_1$
by decomposing terms into a max-top and remaining subterms is doomed to
fail.
Our basic idea for establishing confluence of a (weakly) layered TRS is
to perform rewrite steps on arbitrary terms
on the corresponding elements of a layer system in the terms'
decomposition,
with subterms replaced by variables (this replacement is enabled by \cvar).

Figure~\ref{fig-problems}\subref{fig-problems:bl} depicts the rewrite step
$\m f(\m c,\m c) \to_{\RR_1} \m f(\m c,\m g(\m c))$
with both terms decomposed according to $\bL_1$.
Note that the 
subterm $\m c$
rewrites to $\m g(\m c)$, but the resulting
subterm is split into two layers.
Note furthermore that $\m f(\m c,\m g(\m c)) \to_{\RR_1} \m b$,
but the corresponding left-hand side $\m f(x, \m g(x))$ does not match
any part of the decomposition of $\m f(\m c, \m g(\m c))$.
Condition \crewrite (which is violated by $\bL_1$) helps ensuring
that rewrite steps 
on terms can be adequately simulated on layers.

Next consider Figure~\ref{fig-problems}\subref{fig-problems:fp},
depicting the rewrite step
$\m f(\m c,\m c) \to_{\RR_1} \m f(\m c,\m g(\m c))$
with terms decomposed according to $\bL_2$.
Note that $\bL_2$ satisfies \ctop, \cvar and \crewrite.
Nevertheless,
the result of the rewrite step $\m c \to_{\RR_1} \m g(\m c)$ is broken
apart: only a part of $\m g(\m c)$ is merged with the max-top of
$\m f(\m c,\m g(\m c))$.
Condition \cpartial prevents such partial fusion.
We can see that it is violated by $\bL_2$:
we have $\m f(\square,\m g(\square)) \in \bL_2$ and $\m g(\m c) \in \bL_2$,
but $\m f(\square,\m g(\square) \merge \m g(\m c)) =
\m f(\square, \m g(\m c)) \notin \bL_2$.
Finally, $\bL_3$ weakly layers $\RR$.

In order to motivate \cconsistent, we consider the TRS
$\RR_2$ consisting of the rewrite rules
\begin{xalignat*}{3}
\m{f}(x,x) &\to \m{a} &
\m{f}(x,\m{g}(x)) &\to \m{b} &
\m{h}(\m{c},x) \to \m{g}(\m{h}(x,x))
\end{xalignat*}
which is closely related to $\RR_1$; instead of the rewrite step
$\m c \to_{\RR_1} \m g(\m c)$ we have
$t_c \to_{\RR_2} \m g(t_c)$ for $t_c = \m h(\m c,\m c)$,
and therefore
$\m a \lud{\RR_2}{}\from \m f(t_c,t_c) \to_{\RR_2} \m f(t_c,\m g(t_c))
\to_{\RR_2} \m b$.
We define a layer system $\bL_4$ by
\begin{align*}
\bL_c &= \{ 
v, \m h(v,w), \m h(\m c,v) \mid v,w \in \VC \} \\
\bL_4 &= \{ \m f(\m g^n(s),\m g^m(t)),
\m g^n(t),\m a,\m b,\m c,s \mid s,t \in \bL_c,\, n,m \in \Nat \}
\end{align*}
It is straightforward to verify that $\bL_4$ weakly layers $\RR_2$
and that $\RR_2$ is confluent on $\bLT[4]$.
Figure~\ref{fig-problems}\subref{fig-problems:fa} depicts the
rewrite step $t_c \to_{\RR_2} \m g(t_c)$.
It affects the max-top of $t_c$,
but the max-top of the result, $\m g(\m h(\m c,\square))$,
is larger than the result of rewriting the max-top $\m h(\m c,\square)$
of $t_c$: $\m h(\m c,\square) \to \m g(\m h(\square,\square))$.
In the case of $\RR_2$, there are rewrite sequences in which such
fusion from above happens infinitely often, and that presents another
obstacle to confluence.
Condition \cconsistent is designed to rule out such fusion from above
completely,
and indeed the rewrite step $t_c \to_{\RR_2} \m g(t_c)$
shows that \cconsistent is violated by $\bL_4$ and $\RR_2$.

Finally consider the layer system
\[
\bL_5 = \{ \m f(v,w), \m f(\m g^{n+1}(\m c),\m g^{m+1}(\m c)),\m a,
\m g^n(\m c),\m g^n(v),v \mid v,w \in \VC,\, n,m \in \Nat \}
\]
which weakly layers the TRS $\RR_3$ consisting of the rewrite rules
\begin{xalignat*}{2}
\m{f}(x,x) &\to \m{a} &
\m{c} \to \m{g}(\m{c})
\end{xalignat*}
and satisfies \cconsistent.
Figure~\ref{fig-problems}\subref{fig-problems:ca}
depicts the rewrite step $\m f(\m g(\m c),\m c) \to_{\RR_3}
\m f(\m g(\m c),\m g(\m c))$.
What happens here is that the result of rewriting the subterm
$\m c \to \m g(\m c)$
fuses with the previous top, $\m f(\square, \square)$,
but only if the unrelated first subterm $\m g(\m c)$ fuses at the same
time. This phenomenon causes problems in our proof, and \cstepfusion
prevents that.
To wit, we have $\m f(\square,\square) \in \bL_5$ and
$\m f(\m g(\m c),\m g(\m c)) \in \bL_5$,
so by \cstepfusion with $p = 1$, there should be
$\m f(\square \merge \m g(\m c),\square) \in \bL_5$,
but this is not the case.
\end{example}

The following convention helps to differentiate various contexts.

\begin{convention}
We use $C$ and $D$ to denote \emph{contexts},
$B$ to denote \emph{base contexts} (to be introduced in
Section~\ref{sec-main}), 
$L$ and $N$ to denote arbitrary \emph{layers}, 
and $M$ to denote a max-top of a term or context.
\end{convention}

In the sequel we implicitly assume a given layer system $\bL$.
In light of the next lemma we speak of \emph{the} max-top
of a term or context.

\begin{lemma}
\label{lem-max-top}
Any non-empty context has a unique and non-empty max-top.
\end{lemma}

\begin{proof}
Let $C$ be a non-empty context.
To show that $C$ has a non-empty top let~$x$ be a variable not 
occurring in $C$ and consider $C[x,\ldots,x]$, which has a non-empty
top $L_x$ by \ctop. Then $L := L_x\sigma$ 
with $\dom(\sigma)=\{x\}$ and $\sigma(x)= \square$ is a top of $C$
since $L \in \bL$ by \cvar and $L \Cleq C$ by construction.
It is non-empty since $L = \square$ implies $L_x = x$,
hence $C[x,\ldots,x] = x$ and consequently $C = \square$ because $x$
is fresh, contradicting the premises.
Hence the set $S$ of non-empty tops of $C$ is non-empty. Since it
also is finite it has a (non-empty) maximal element.

To show uniqueness let $M$ and $M'$ be max-tops of $C$.
Then $M \Cleq C$ and $M' \Cleq C$ ensures that $M \merge M'$ is defined,
and a layer by \cpartial
(if $\square \in \{ M,M' \}$ then \cpartial is not needed).
If $M \neq M'$ then $M \Clt M\merge M'$ or $M' \Clt M\merge M'$. Since
$M\merge M' \Cleq C$ this gives the desired contradiction.
\end{proof}

Next we introduce the key notion of the rank of a term.

\begin{definition}
\label{def-rank}
Let $t = M[\seq{t}]$ with $M$ the max-top of $t$.
Then $\seq{t}$ are the \emph{aliens} of $t$.
We define
$\rank(t) = 1 + \max\{ \rank(t_i) \mid  1 \leq i \leq n \}$,
where $\max(\varnothing) = 0$ by convention.
\end{definition}

Since the max-top of a term is uniquely defined (Lemma~\ref{lem-max-top}),
it follows that also its aliens are uniquely defined.
The next example shows that rewriting might increase the rank of a term.
In Lemma~\ref{lem-rank-decrease} 
we show that this cannot happen in weakly layered TRSs.

\begin{example}
\label{ex-rank}
Consider the layer system
\[
\bL_6 = \{ v, \m{f}(v), \m{g}(v), \m{h}(v), \m{f}(\m{g}(\m{h}(v))),
\m{g}(\m{g}(v)), \m a \mid v \in \VC \}
\]
Note that (in contrast to modularity) subterms can have larger rank.
E.g., if $s = \m{f}(\m{g}(\m{h}(x)))$ and $t = \m{g}(\m{h}(x))$
then $\rank(s) = 1 < 2 = \rank(t)$.
Furthermore $s \to_\RR t$ in the TRS $\RR$ containing the rule
$\m{f}(\m{g}(x)) \to \m{g}(x)$.
Note that $\RR$ is not weakly layered according to $\bL_6$.
\end{example}

The next lemma states a useful decomposition result.

\begin{lemma}
\label{lem-fusion}
Let $t = L[\seq{t}]$, $L$ a top of $t$,
and $k$ be the maximum of $\rank(t_i)$ for $1 \leq i \leq n$.
Then $\rank(t) \leq k + 1$ and aliens of~$t$ that are not rooted 
at hole positions of $L$ have rank less than~$k$.
\end{lemma}

\begin{proof}
Let $M$ be the max-top of $t$.
We show the (stronger) property for any context $C$ with $C \Cleq M$ 
(instead of a top $L$ of $t$).  Note that $L \Cleq M$. 
The proof is by induction on $|t| - |C|_{\FF \cup \VV}$,
which is a natural number because $C \Cleq t$. If $C = M$ then
$\rank(t) = 1 + \max\{ \rank(t_i) \mid 1 \leq i \leq n \} = 1 + k$
and all aliens of $t$ are rooted at hole positions of $C$, so we are done.
Otherwise, let $M_i$ be the max-top of $t_i$. There is a unique maximal
context $C'$ such that $C' \Cleq C[M_1,\dots,M_n]$ and $C' \Cleq M$.
Furthermore, we have $C \Clt C'$ because the $M_i$ are non-empty
by Lemma~\ref{lem-max-top}.
Because $C' \Cleq M \Cleq t$,
$t = C'[t_1',\dots,t_m']$
where $t_j'$ is the subterm of $t$ at the position of the $j$-th hole
in $C'$.
For each $p \in \Pos_\square(C)$ there are three possibilities.
Let $C[\seq{t}]|_p = t_i$.
\begin{enumerate}
\item
If $p \in \Pos_\square(C')$ then
$C'[\seq[m]{t'}]|_p = t_j'$ and $t_i = t_j'$ for some $j$.
\item
If $p \in \Pos_\VV(C')$ then there are no holes below $p$ in $C'$.
\item
\label{itm-lem-fusion-3}
If $p \in \Pos_\FF(C')$ then $p \in \Pos_\FF(M)$
and $M[M|_p \merge M_i]_p \in \bL$ by \cpartial. Because
$M$ is the max-top of $t$ this implies $M_i \Cleq M|_p$ and
therefore $C'|_p = M_i$ by construction of $C'$.
Hence all $t_j'$ corresponding to holes of $C'$ below
$p$ are aliens of $t_i$ having rank less than $\rank(t_i)$.
\end{enumerate}
We can now apply the induction hypothesis to $C'[t_1',\dots,t_m']$
since $C \Clt C'$ implies
$|t| - |C|_{\FF \cup \VV} > |t| - |C'|_{\FF \cup \VV}$.
To conclude, note that any alien rooted
at a hole position of $C'$ but not at a hole position of $C$
equals a $t_j'$ from case (\ref{itm-lem-fusion-3}) and therefore
has rank less than $k$.
\end{proof}

\begin{lemma}
\label{lem-bl-closed}
Let $\RR$ be a TRS that is weakly layered according to $\bL$.
Then $\bL$ is closed under rewriting by $\RR$.
\end{lemma}

\begin{proof}
Let $L \in \bL$ and $L \to_\RR N$. 
Obviously $L[x,\ldots,x] \to_\RR N[x,\ldots,x]$ for a fresh variable $x$.
Since $L[x,\ldots,x] \in \bL$ by \cvar it is its own max-top.
We conclude since $N[x,\ldots,x] \in \bL$ by \crewrite and
hence $N \in \bL$ by \cvar.
\end{proof}

We now present technical results about rewriting contexts.
In the sequel we often want to replace variables affected by
some substitution $\sigma$ by holes. We therefore denote by
$\sigma_\square(x)$ the substitution obtained by letting
$\sigma_\square(x) = \square$ for $x \in \dom(\sigma)$ and
$\sigma_\square(x) = x$ otherwise. For a context $C$
we denote by $C_\square$ the context obtained from 
$C$ by replacing all variables by holes.

\begin{lemma}
\label{lem-context-instance}
Let $C$ be a context and $\ell$ a non-variable term.
If $\ell \matches C|_p$ then there is a term $c$ such that 
\begin{enumerate}
\item 
\label{lem-context-instance-C}
$\ell \matches c|_p$ and $C = c\sigma_\square$ for some 
substitution $\sigma$, and
\item 
\label{lem-context-instance-D}
if $C \Cleq D$ for a context $D$ and $\ell \matches D|_p$ then $c \matches D$.
\end{enumerate}
\end{lemma}

\begin{proof}
Assume that $C$ has $n \geqslant 0$ holes.
We may assume without loss of generality that $C$ and $\ell$
have no variables in common.
Let $c_0 := C[\seq{x}]$ with fresh variables $\seq{x}$.
The context $C$ witnesses the fact that $c_0$ and $c_1 := c_0[\ell]_p$
are unifiable.
Let $c$ be a most general instance of $c_0$ and $c_1$.
Note that variables in $c$ can be renamed freely.
If $C \Cleq D$ then $D$ is an instance of $c_0$. Furthermore,
if $\ell \matches D|_p$ then $D$ must be an instance of $c_1$ as well
and therefore $c \matches D$. In particular, $c \matches C$
and thus $C = c\sigma$ for some substitution $\sigma$.
Let $\tau$ be a substitution such that $c = c_0\tau$.
For $x \in \Var(C)$, $\sigma(\tau(x)) = x$, which implies that $\tau(x)$
is a variable.
We can rename each $\tau(x)$ to $x$ in $c$.
Therefore we may assume without loss of generality that
$\sigma(x) = \tau(x) = x$ for $x \in \Var(C)$.
For the variables $x_i$,
we have $\sigma(\tau(x_i)) = \square$ for all $1 \leq i \leq n$,
which is only possible if $\sigma$ maps those variables to $\square$.
Consequently, $\sigma_\square = \sigma$.
\end{proof}

If a rewrite rule is applied to a context then each hole may be
erased, copied or duplicated. The same holds for the terms used to fill the
holes in a context, as formalized by the next lemma.

\begin{lemma}
\label{lem-context-rewrite}
If $C \to_{p,\ell \to r} C'$ and $\ell \matches C[\seq{s}]|_p$
then $C[\seq{s}] \to_{p,\ell \to r} C'[\seq[m]{t}]$ and
$\{ \seq[m]{t} \} \subseteq \{ \seq{s} \}$.
\end{lemma}
\begin{proof}
Since $\ell \matches C|_p$, 
Lemma~\ref{lem-context-instance}\eqref{lem-context-instance-C} yields
a term $c$ and a substitution
$\sigma_\square$ such that
$\ell \matches c|_p$ and $C = c\sigma_\square$.
Furthermore due to $C \Cleq C[\seq{s}]$ and
$\ell \matches C[\seq{s}]|_p$,
there is a substitution~$\sigma$ with $c\sigma = C[\seq{s}]$ by
Lemma~\ref{lem-context-instance}\eqref{lem-context-instance-D}.
Hence $C \to_{p,\ell\to r} C'$ mirrors
a rewrite step $c \to_{p,\ell \to r} c'$
with $C' = c'\sigma_\square$ and $C'[\seq[m]{t}] = c'\sigma$.
Since $\seq[m]{t}$ can only come from $\sigma$ we conclude.
\end{proof}

This section ends with a key lemma that enables the use
of induction on the rank of terms for proving confluence of $\RR$.

\begin{lemma}
\label{lem-rank-decrease}
Let $\RR$ be a weakly layered TRS.
If $s \to_\RR t$ then $\rank(s) \geq \rank(t)$.
\end{lemma}

\begin{proof}
By induction on the rank of $s$. Let $s \to_p t$ and $s = M[\seq{s}]$ be
the decomposition of $s$ into max-top and aliens. We distinguish two
cases.

If $p \in \Pos_\FF(M)$ then condition \crewrite yields $M \to_p L$ and
$L$ a top of $t$. Let $t = L[\seq[m]{t}]$. 
By Lemma~\ref{lem-context-rewrite} 
$\{ \seq[m]{t} \} \subseteq \{ \seq{s} \}$ since $M \to_p L$. Hence
$\rank(t) \leq 1 + \max\{ \rank(t_i) \mid 1 \leq i \leq m \}
\leq 1 + \max\{ \rank(s_i) \mid 1 \leq i \leq n \} = \rank(s)$
using Lemma~\ref{lem-fusion}.

If $p \notin \Pos_\FF(M)$ then $s_j \to s_j'$ and 
$t = M[s_1,\dots,s_j',\dots,s_n]$ for some ${1 \leq j \leq n}$. 
The induction hypothesis yields $\rank(s_j) \geq \rank(s_j')$. 
Since $M$ is a top of $t$, Lemma~\ref{lem-fusion} yields
$\rank(t) \leq 
 1 + \max\{ \rank(s_j'), \rank(s_i) \mid 1 \leq i \leq n, i \neq j \}
\leq 1 + \max\{ \rank(s_i) \mid {1 \leq i \leq n} \} = \rank(s)$.
\end{proof}

\section{Confluence by Layer Systems}
\label{sec-main}

We start this long section by stating our main results. All results
reduce the task of proving confluence of a TRS to the easier task of
proving confluence of the terms in a suitable layer system, i.e., 
the terms in $\bLT$, which are precisely the terms of rank one.
The first result imposes left-linearity.

\begin{theorem}
\label{thm-main-ll}
Let $\RR$ be a weakly layered TRS that is confluent on terms of rank one.
If $\RR$ is left-linear then $\RR$ is confluent.
\end{theorem}

The second result exchanges left-linearity for a condition that
is weaker than non-duplication.

\begin{definition}
\label{def-bd}
Let $\RR$ be a TRS and $\Diamond$ a fresh unary function symbol. Then
$\RR$ is \emph{bounded duplicating} if the relative rewrite system
$\{ \Diamond(x) \to x \} / \RR$ is terminating.
\end{definition}

\begin{theorem}
\label{thm-main-bd}
Let $\RR$ be a weakly layered TRS that is confluent on terms of rank one.
If $\RR$ is bounded duplicating then $\RR$ is confluent.
\end{theorem}

\begin{lemma}
\label{lem-nd-bd}
Non-duplicating TRSs are bounded duplicating.
\end{lemma}

\begin{proof}
Let $\RR$ be a non-duplicating TRS.
In order to show termination of
$\{ \Diamond(x) \to x \} / \RR$ we measure
terms by counting the number of occurrences of the $\Diamond$ symbol.
Clearly each application of the $\Diamond(x) \to x$ rule decreases
that number and rules of $\RR$ do not increase it because they
do not duplicate $\Diamond$ symbols and cannot introduce any
new ones.
\end{proof}

Bounded duplication strictly extends non-duplication; the TRS
consisting of the rewrite rule
$\m{f}(x,x) \to \m{g}(x,x,x)$ is duplicating but still
bounded duplicating.
This can be shown by the polynomial interpretation~\cite{L79}
given by
\begin{xalignat*}{3}
\m{f}_\Nat(x,y) &= 2x + 2y&
\m{g}_\Nat(x,y,z) &= x+y+z &
\Diamond_\Nat(x) &= x+1
\end{xalignat*}
By combining Theorem~\ref{thm-main-bd} with Lemma~\ref{lem-nd-bd}, we
obtain the following corollary.

\begin{corollary}
\label{thm-main-nd}
Let $\RR$ be a weakly layered TRS that is confluent on terms of rank one.
If $\RR$ is non-duplicating then $\RR$ is confluent.
\qed
\end{corollary}

The third result does not impose any conditions on $\RR$ but further
limits the layer systems that can be employed to derive confluence.

\begin{theorem}
\label{thm-main-con}
Let $\RR$ be a layered TRS that is confluent on terms of rank one.
Then $\RR$ is confluent.
\end{theorem}

Hence for duplicating TRSs there are three possibilities to prove
confluence, either by weakly layering a left-linear rewrite system
(Theorem~\ref{thm-main-ll}),
by establishing bounded duplication for a weakly layered rewrite system
(Theorem~\ref{thm-main-bd}),
or by layering the rewrite system (Theorem~\ref{thm-main-con}).
Table~\ref{tab:incomparable} shows that the three results are pairwise
incomparable where 
$\bL_7 = \{ v, \m{k}(v,w), \m{b} \mid v, w \in \VC \}$
and $\bL_6$ is as in Example~\ref{ex-rank}.
\begin{table}[tb]
\begin{center}
\[
\renewcommand{\arraystretch}{1.2}
\begin{array}{@{}rclc@{~}|@{~}ccc@{}}
\multicolumn{3}{c}{\text{rewrite rule}} & \text{layer system} &
\text{Theorem~\ref{thm-main-ll}}
& \text{Theorem~\ref{thm-main-bd}} & \text{Theorem~\ref{thm-main-con}} \\
\hline
\m{f}(\m{g}(\m{h}(x))) &\to& \m{g}(x) & \bL_6
& \SUCC & \SUCC & \FAIL \\
\m{k}(\m{b},x) &\to& \m{k}(x,x) & \bL_7
& \SUCC & \FAIL & \SUCC \\
\m{k}(x,x) &\to& \m{k}(x,x) & \bL_7
& \FAIL & \SUCC & \SUCC
\end{array}
\renewcommand{\arraystretch}{1}
\]
\caption{Incomparability of the main results.}
\label{tab:incomparable}
\end{center}
\end{table}

In the following subsections we develop proofs for
Theorems~\ref{thm-main-ll}, \ref{thm-main-bd}, and~\ref{thm-main-con}.
In Section~\ref{sec-setup} we describe the proof setup and
introduce auxiliary rewrite relations. In Sections~\ref{sec-dd-tall}
and \ref{sec-dd-short} we show that the auxiliary relations are
locally decreasing. Finally, we wrap up the proofs in
Section~\ref{sec-dd-proof}.

\subsection{Proof Setup}
\label{sec-setup}

Assume we are given a weakly layered TRS $\RR$ such that $\RR$ is
confluent on terms of rank one. We will show confluence of $\RR$ on all
terms by induction on the rank of terms.
In the sequel we prepare for the induction step, hence:
\begin{center}
\fbox{\emph{We fix $r$ and assume terms with rank at most $r$ to be
confluent.}}
\end{center}

Next we generalize the crucial concepts of \cite{vO08b}
from the modularity setting to layer systems.
We have renamed \emph{non-native} to \emph{foreign} because
non-native is not the complement of native.

\begin{definition}
Terms with rank at most $r+1$ are called \emph{native}.
An alien of a native term is \emph{tall}
if its rank equals $r$ and \emph{short} otherwise.
\emph{Foreign} terms have rank less than or equal to $r$.
\end{definition}

Note that by definition, foreign terms are also native.
However, we will only call terms foreign if they are descendants of 
aliens of a native term.

\begin{definition}
\label{def-steps}
Let $t$ be a native term. Its \emph{base context} $B$ is obtained by
replacing
all tall aliens in $t$ with holes. The tall aliens form the
\emph{base sequence} $\myvec t$, which satisfies $t = B[\myvec t]$.
\end{definition}

\begin{definition}
Sequences of foreign terms are called \emph{foreign sequences}.
The \emph{imbalance} of a foreign sequence $\myvec t$
is the number of distinct terms in $\myvec t$.
The imbalance of a native term $t$ is the imbalance of its base sequence.
If $\myvec s$ and $\myvec t$ are sequences of length $n$, then
we write $\myvec s \propto \myvec t$ if $s_i = s_j$ implies
$t_i = t_j$ for all $1 \leq i, j \leq n$.
\end{definition}

Note that the relation $\propto$ is transitive. It is useful for
analyzing the imbalance of foreign sequences. If
$\myvec s \propto \myvec t$ then the imbalance of
$\myvec t$ is no larger than that of $\myvec s$.

\begin{definition}
\label{short steps}
Let $s$ and $t$ be native terms.
A \emph{short step} $s \RRR[s_0]{} t$ is a sequence of $\RR$-steps
$s \to_\RR^* t$
that is mirrored by a rewrite sequence $B \to_\RR^* C$
from the base context $B$ of $s$.
Short steps are labeled by terms $s_0$ that are predecessors of
the source: $s_0 \to_\RR^* s$. We omit the label when it is irrelevant.
\end{definition}

There are two ways in which short steps arise: either by rewriting short
aliens
(hence the name), or by rewriting the max-top of a term.
In the sequel we will sometimes use the fact that in
Definition~\ref{short steps}, $C \Cleq t$ by
Lemma~\ref{lem-context-rewrite},
and when writing $s = B[\myvec s]$ and $t = C[\myvec t]$, each element of
$\myvec t$ is an element of $\myvec s$.

\begin{definition}
\label{tall steps}
Let $B$ and $\myvec s$ be the base context and base sequence of a
native term~$s$. If $\myvec s \to_\RR^* \myvec t$ then
$s = B[\myvec s] \rrr[\iota]{} B[\myvec t] = t$ is a \emph{tall step}.
Here the label $\iota$ is the imbalance of $\myvec t$.
\end{definition}

Note that $\myvec t$ 
in Definition~\ref{tall steps}
is a foreign sequence because $\RR$ is weakly layered.
Further note that the imbalance of $t$ 
may be smaller than $\iota$ (since $B$ need not be the base context
of $t$).
The following example illustrates the above concepts.

\begin{example}
\label{ex-vO08b}
Consider the TRSs
$\RR_1 = \{ \m{f}(x,x) \to x \}$ over
$\FF_1 = \{ \m{f}, \m{a} \}$ and 
$\RR_2 = \{ \m{G}(x) \to \m{I}, \m{I} \to \m{K}, \m{G}(x) \to \m{H}(x), 
\m{H}(x) \to \m{J}, \m{J} \to \m{K} \}$ over
$\FF_2 = \{ \m{I}, \m{J}, \m{K}, \m{G}, \m{H} \}$
and let $\RR = \RR_1 \cup \RR_2$.
Then $\bL = \CC(\FF_1,\VV) \cup \CC(\FF_2,\VV)$
layers $\RR_1$ and $\RR_2$ (cf.\ the proof of
Theorem~\ref{thm-modularity}).
Assume that $r = 2$.
Table~\ref{tab-native} demonstrates some properties and notions.
We have 
\newcommand\ulalien[1]{\underline{\vphantom{+}\smash{#1}}}
${\m{f}(\ulalien{\m{G}(\m{a})},\ulalien{\m{G}(\m{a})}) \RRR[]{} \m{G}(\m{a})}$ but
$\m{f}(\ulalien{\m{G}(\m{a})},\ulalien{\m{G}(\m{a})}) \RRR[]{} \m{I}$ is not possible 
since the step $\m{G}(\m{a}) \to_\RR \m{I}$ is not in the base
context of $\m{f}(\m{G}(\m{a}),\m{G}(\m{a}))$.
(Here we have underlined the tall aliens.) We also have 
$\m{f}(\ulalien{\m{G}(\m{a})},\ulalien{\m{G}(\m{a})}) \rrr[2]{} 
 \m{f}(\m{J},\m{G}(\m{a})) = u$,
despite the imbalance of $u$ being 1 (note that $\m{f}(\square,\square)$
is not the base context of~$u$). 
Furthermore,
$(\m{G}(\m{a}),\m{G}(\m{a})) \not\propto (\m{J},\m{G}(\m{a}))$ but
as the latter can be further rewritten
$(\m{J},\m{G}(\m{a})) \to^*_\RR (\m{J},\m{J})$
we obtain $(\m{G}(\m{a}),\m{G}(\m{a})) \propto (\m{J},\m{J})$.
\end{example}

\begin{table}[tb]
\begin{center}
\renewcommand{\arraystretch}{1.2}
\begin{tabular}{@{}r@{~$=$~}l|cclllc@{}}
\multicolumn{2}{c|}{term} & foreign
& native &
max-top & base context & base sequence & imbalance \\
\hline
$s$ & $\m{f}(\m{G}(\m{a}),\m{G}(\m{a}))$ &
$\FAIL$ &
$\SUCC$ &
$\m{f}(\square,\square)$ &
$\m{f}(\square,\square)$ &
$(\m{G}(\m{a}),\m{G}(\m{a}))$ &
1
\\
$t$ & $\m{f}(\m{H}(\m{a}),\m{G}(\m{a}))$ &
$\FAIL$ &
$\SUCC$ &
$\m{f}(\square,\square)$ &
$\m{f}(\square,\square)$ &
$(\m{H}(\m{a}),\m{G}(\m{a}))$ &
2
\\
$u$ & $\m{f}(\m{J},\m{G}(\m{a}))$ &
$\FAIL$ &
$\SUCC$ &
$\m{f}(\square,\square)$ &
$\m{f}(\m{J},\square)$ &
$(\m{G}(\m{a}))$ &
1
\\
$v$ & $\m{f}(\m{K},\m{K})$ &
$\SUCC$ &
$\SUCC$ &
$\m{f}(\square,\square)$ &
$\m{f}(\m{K},\m{K})$ &
$()$
& 0
\end{tabular}
\renewcommand{\arraystretch}{1}
\caption{Properties for $r = 2$.}
\label{tab-native}
\end{center}
\end{table}

\begin{remark}
The constraint on short steps is subtle. It implies that the rewrite
steps do not overlap with any descendants of the tall aliens of $s$, but
furthermore it also has the effect of delaying fusion of those tall
aliens with the base context until the end of the rewrite sequence,
in the sense of \cite{FZM11}.
\end{remark}

We prove confluence of $\RR$ on native terms
by showing that any local peak consisting of short steps and/or tall steps
may be joined decreasingly.
Steps are compared as follows.
Tall steps are ordered by their imbalance,
tall steps are ordered above short steps, and
short steps are compared by a well-founded order introduced later
(in the proof of Lemma~\ref{lem-ss-bd}).

In the remainder of this section we use
$s$, $t$, and $u$ to denote native terms.

\subsection{Local Decreasingness of Peaks involving Tall Steps}
\label{sec-dd-tall}

Based on Lemma~\ref{lem-context-rewrite} we obtain the following result:

\begin{lemma}
\label{lem-balanced-rewrite}
Let $\myvec s$ and $\myvec t$ be sequences of contexts with
$\myvec s \propto \myvec t$ and
$C \to_{p,\ell \to r} C'$. If $\ell \matches C[\myvec s]|_p$
then $C[\myvec t] \to_{p,\ell \to r} C'[\myvec t']$ with each element of
$\myvec t'$ belonging to $\myvec t$.
\end{lemma}

\begin{proof}
We extend the proof of Lemma~\ref{lem-context-rewrite} as follows.
Let $\tau$ be the substitution
\[
\tau(x) = \begin{cases}
t_i & \text{if $x \in \dom(\sigma_\square)$ and $\sigma(x) = s_i$} \\
x & \text{otherwise}
\end{cases}
\]
Note that $C[\myvec t] = c\tau$ because $\myvec s \propto \myvec t$.
We have $c\tau \to_{p,\ell \to r} c'\tau$. Comparing
$c'\tau$ and $C' = c'\sigma_\square$ establishes the claim that
$c'\tau = C'[\myvec t']$ with each element of $\myvec t'$
equaling some element of $\myvec t$.
\end{proof}

\begin{lemma}
\label{lem-balance}
Let $\myvec s$, $\myvec t$, $\myvec u$ be foreign sequences.
If $\myvec s \to_\RR^* \myvec t$ and $\myvec s \to_\RR^* \myvec u$ then
there is a foreign sequence $\myvec v$ such that
$\myvec t \to_\RR^* \myvec v$,
$\myvec u \to_\RR^* \myvec v$ with
$\myvec t \propto \myvec v$ and
$\myvec u \propto \myvec v$.
\end{lemma}

\begin{proof}
Let $m$ be the length of $\myvec s$.
We use induction on the number of disequalities
$t_i \neq u_i$ for $1 \leq i \leq m$.
If this number is zero then $\myvec t = \myvec u$ and we can
take $\myvec v = \myvec t$.
Otherwise, 
$t_i \neq u_i$ for some $1 \leq i \leq m$.
Both $t_i$ and $u_i$ are reducts of $s_i$ and thus have a common reduct $v$
since $\RR$ is confluent on foreign terms. By replacing every
occurrence of $t_i$ and $u_i$ in 
$\myvec t$, $\myvec u$ by $v$, we obtain new
sequences $\myvec t'$, $\myvec u'$ that satisfy
$\myvec s \to_\RR^* \myvec t \to_\RR^* \myvec t'$,
$\myvec s \to_\RR^* \myvec u \to_\RR^* \myvec u'$,
$\myvec t \propto \myvec t'$
and
$\myvec u \propto \myvec u'$. Since the 
number of disequalities $t_i' \neq u_i'$ is decreased, we
conclude by the induction hypothesis and the transitivity of $\propto$.
\end{proof}

A step in the base context is short.

\begin{lemma}
\label{lem-short-step}
Let $p$ be a non-hole position of the base context of $s$.
If $s \to_p t$ then $s \RRR{} t$.
\end{lemma}

\begin{proof}
Let $B$ be the base context of $s$ and let $s \to_p t$.
We show $B \to_p C$ for some context $C$.
Because left-hand sides of rules are not variables, $p \in \Pos_\FF(B)$.
Let $M$ be the max-top of $s$, which is also the max-top of $B$.
We distinguish two cases.
If $p \in \Pos_\FF(M)$ then consider the decomposition $s = M[\myvec s]$.
According to \crewrite there is a layer $L$ with $M \to_{p} L$.
We have $B = M[\myvec s']$ where $s'_i = s_i$ if $s_i$ is
a short alien and $s_i' = \square$ if $s_i$ is tall.
Clearly $\myvec s \propto \myvec s'$ and hence we conclude by
Lemma~\ref{lem-balanced-rewrite}.
If $p \notin \Pos_\FF(M)$ then $s|_p$ is a subterm of a short alien
of $s$ and thus $B|_p = s|_p$. Hence
$B \to_{p} C$ for the context $C := B[t|_p]_p$.
\end{proof}

When doing a short step $s = B[\myvec s] \RRR{} C[\myvec s'] = t$,
in general the context $C$ is not the base context of $t$ (because of
fusion from above or conspiring aliens).
Similarly, for a tall step $s = B[\myvec s] \rrr{} B[\myvec t] = t$
in general the context $B$ is not the base context of $t$ (because of
fusion caused by steps in the aliens of $t$),
but both contexts ($B$ and $C$) satisfy the 
more general property defined below.

\begin{definition}
We call a context \emph{shallow} if its rank is at most $r$ and
all its aliens are terms from $\TT(\FF,\VV)$.
\end{definition}

Note that the base contexts of native terms are shallow.
The same holds for the max-tops of native terms.
Furthermore, shallow contexts are closed under rewriting,
as shown by the next lemma.

\begin{lemma}
\label{lem-shallow-closed}
If $C$ is a shallow context and $C \to_\RR D$ then $D$ is
a shallow context.
\end{lemma}

\begin{proof}
Assume that $C \to_{p,\ell \to r} D$.
Then $C[x,\dots,x] \to_{p,\ell \to r} D[x,\dots,x]$ for a fresh
variable $x$. Let $M_x$ be the max-top of $C[x,\dots,x]$ and
note that the max-top $M$ of $C$ is obtained by replacing each
occurrence of $x$ by a hole in $M_x$.
If $p \in \Pos_\FF(M) = \Pos_\FF(M_x)$ then by \crewrite there is a
rewrite step $M_x \to_{p,\ell \to r} L_x$ where $L_x$ is a layer,
and even a top of $D[x,\dots,x]$ by Lemma~\ref{lem-context-rewrite}.
There is a mirroring rewrite step $M \to_{p,\ell \to r} L$
where $L$ is a top of $D$. By Lemma~\ref{lem-context-rewrite},
each hole of $L$ corresponds to a hole or a term without holes in $D$.
If $p \notin \Pos_\FF(M)$ then we take $L = M$, which is a top
of $D$. Again, each hole of $L$ corresponds to a hole or a term in $D$.
In both cases we conclude by noting that any holes of $D$ are
holes of $L$ and therefore also of the max-top of $D$ and that
the rank of $D$, which equals the rank of $D_x$, is at most $r$ by
Lemma~\ref{lem-rank-decrease}.
\end{proof}

Let $s = B[\myvec s]$ be the decomposition of $s$ into base context and
base sequence. From the previous result we get that 
$B[\myvec s] \RRR{} C[\myvec s'] = t$ (with $B \to^*_\RR C$) implies
that $C$ is shallow.
The next result establishes that the shallow context $C$ is never larger
than the base context of $t$.

\begin{lemma}
\label{lem-base-increase}
Let $C$ be a shallow context and $t$ a native term.
If $C \Cleq t$ then $C \Cleq B$ for the base context $B$ of $t$.
\end{lemma}

\begin{proof}
Let $C = M[\myvec s]$ be the decomposition of $C$ into max-top and
aliens. Since $C$ is shallow, elements of $\myvec s$ are either holes or
terms of rank less than~$r$. From $M \Cleq C \Cleq t$ we infer the
existence of a sequence $\myvec t'$ such that $t = M[\myvec t']$ and
$s_i = t'_i$ whenever $s_i \neq \square$.
By Lemma~\ref{lem-fusion} every tall alien in $t$
is a subterm of a term of rank at least $r$ in
$\myvec t'$. Hence $C \Cleq B$ as desired.
\end{proof}

Steps within shallow contexts are short steps.
\begin{lemma}
\label{lem-shallow-step}
Let $p$ be a non-hole position in a shallow context $C$ with
$s = C[\myvec s]$. If $s \to_p t$ then $s \RRR{} t$.
\end{lemma}
\begin{proof}
By Lemmata~\ref{lem-base-increase} and~\ref{lem-short-step}.
\end{proof}

Steps below a shallow context can be decomposed into tall and short steps.

\begin{lemma}[(tall--short factorization)]
\label{lem-tall-short}
Let $s = C[\myvec s]$ with a shallow context $C$
and a foreign sequence $\myvec s$.
If $\myvec s \to_\RR^* \myvec t$ and $\iota$ is the imbalance of
$\myvec t$ then
$C[\myvec s] \rrr[\leq\iota]{} {\cdot} \RRR{*} C[\myvec t]$.
\end{lemma}

\begin{proof}
Let $B$ and $\myvec s'$ be the base context and base sequence of $s$.
Note that by Lemma~\ref{lem-fusion} (with $L$ equal to the max-top of
$C$) the tall aliens $\myvec s'$ of $s$ are a subsequence of $\myvec s$,
because all aliens of $C$ have rank less than $r$.
For the corresponding subsequence $\myvec t'$ of $\myvec t$,
we obtain $s = B[\myvec s'] \rrr[\leq\iota]{} B[\myvec t']$,
while the remaining elements of $\myvec s$ and $\myvec t$
give rise to a rewrite sequence
$B = C[\myvec s''] \to_\RR^* C[\myvec t'']$,
where $\myvec s''$ ($\myvec t''$) is obtained by replacing
the terms corresponding to the elements of $\myvec s'$ ($\myvec t'$)
by holes.
Consequently,
$B[\myvec t'] = C[\myvec s''][\myvec t'] \RRR{*} C[\myvec t''][\myvec t']
= C[\myvec t]$
by Lemma~\ref{lem-shallow-step}.
\end{proof}

\begin{example}
Continuing Example~\ref{ex-vO08b}.
Let $s = \m f(\m J,\m G(\m a))$. Then $s = C[\myvec s]$ for the shallow
context
$C = \m f(\square,\square)$ with $\myvec s = (\m J, \m G(\m a))$.
Let $\myvec t = (\m K, \m I)$. Since $\myvec s \to_\RR^* \myvec t$
the conditions of Lemma~\ref{lem-tall-short} hold and we have
$C[\myvec s] \rrr[\leq 2]{} \cdot \RRR[]{*} C[\myvec t]$.
The tall step arises as
$s = \m f(\m J,\square)[\m G(\m a)] \rrr[1]{} 
     \m f(\m J,\square)[\m I] = \m f(\m J,\m I)$
while $\m f(\m J,\m I) \RRR[]{} \m f(\m K, \m I)$
is a short step since $\m f(\m J,\m I)$ is its own base context.
\end{example}

\begin{lemma}
\label{lem-peak-tt}
Local peaks of tall steps are decreasing:
\[
{\lll[\iota]{}} \cdot {\rrr[\kappa]{}} ~\subseteq~
{\rrr[\leq\kappa]{}} \cdot {\RRR{*}} \cdot {\LLL{*}}
\cdot {\lll[\leq\iota]{}}
\]
\end{lemma}

\begin{proof}
Let $t \lll[\iota]{} s \rrr[\kappa]{} u$ and let the base context
and base sequence of $s$ be $B$ and $\myvec s$. There are
foreign sequences $\myvec t$ and $\myvec u$ such that
$\myvec t \lud{\RR}{*}\from \myvec s \to_\RR^* \myvec u$ and
$t = B[\myvec t]$, $u = B[\myvec u]$.
By Lemma~\ref{lem-balance}, we can find a foreign sequence
$\myvec v$ such that
$\myvec t \to_\RR^* \myvec v \lud{\RR}{*}\from \myvec u$,
$\myvec t \propto \myvec v$, and $\myvec u \propto \myvec v$. Hence
the imbalance of $\myvec v$ is less than or equal to both $\iota$
and $\kappa$ and we conclude by Lemma~\ref{lem-tall-short}.
\end{proof}

\begin{example}
To demonstrate Lemma~\ref{lem-peak-tt}, we extend Example~\ref{ex-vO08b}.
Let $s = \m f(\m G(\m a), \m G(\m a))$. Then
$t = \m f(\m H(\m a), \m I) \lll[2]{} s \rrr[2]{}
\m f(\m I,\m H(\m a)) = u$.
Note that $\m I \to_\RR \m K$ and
$\m H(\m a) \to_\RR \m J \to_\RR \m K$.
The base contexts of $t$ and $u$ are $\m f(\square,\m I)$ and
$\m f(\m I,\square)$, respectively.
Consequently, $t \rrr[1]{} \m f(\m K,\m I) \RRR{} \m f(\m K,\m K)
\LLL{} \m f(\m I, \m K) \lll[1]{} u$.
\end{example}

\begin{lemma}
\label{lem-peak-ts}
Local peaks involving a tall and a short step are decreasing:
\[
{\lll[\iota]{}} \cdot {\RRR{}} ~\subseteq~ 
{\rrr[< \iota]{=}} \cdot {\RRR{*}} \cdot {\LLL{*}} \cdot
{\lll[\leq \iota]{}}
\]
\end{lemma}

\begin{proof}
Let $t \lll[\iota]{} s \RRR{} u$ and let the base context and
base sequence of $s$ be $B$ and $\myvec s$. We have $t = B[\myvec t]$
with $\myvec s \to_\RR^* \myvec t$ for some foreign sequence $\myvec t$
and $u = C[\myvec u]$.
We construct $\myvec v$ and $\myvec w$ such that
${B[\myvec t] \rrr[< \iota]{=} \cdot \RRR{*} B[\myvec v] \RRR{*} 
C[\myvec w]  \LLL{*} \cdot \lll[\leq \iota]{} C[\myvec u]}$.
We distinguish two cases.
\begin{enumerate}
\item
If $\myvec s \propto \myvec t$ then we let $\myvec v = \myvec t$.
Hence $B[\myvec t] = B[\myvec v]$ and thus
$B[\myvec t] \rrr[< \iota]{=} {\cdot} \RRR{*} B[\myvec v]$.
\item
Otherwise, using Lemma~\ref{lem-balance} with
$\myvec s \to_\RR^* \myvec t$
and $\myvec s \to_\RR^* \myvec s$
we can find a foreign sequence $\myvec v$ such that
$\myvec t \to_\RR^* \myvec v$, $\myvec t \propto \myvec v$,
and $\myvec s \propto \myvec v$. Since the imbalance of $\myvec v$
is less than $\iota$
($\myvec s \not\propto \myvec t$ means that there are $i,j$ with
$s_i = s_j$ and $t_i \neq t_j$.
By $\myvec s \propto \myvec v$, we have $v_i = v_j$,
and $\myvec t \propto \myvec v$ ensures that
all other equalities between elements of $\myvec t$ carry over to
$\myvec v$,
so the imbalance becomes smaller)
we obtain $B[\myvec t] \rrr[< \iota]{=} {\cdot} \RRR{*} B[\myvec v]$
from Lemma~\ref{lem-tall-short}.
\end{enumerate}
By the definition of $\RRR{}$ we get $B \to_\RR^* C$
mirroring
$s = B[\myvec s] \to_\RR^* C[\myvec u] = u$.
Hence $\myvec u$ is a sequence of foreign terms such that all elements
of $\myvec u$ are elements of $\myvec s$, which follows by repeated
application of Lemma~\ref{lem-context-rewrite}. We define
$w_i = v_j$ if $u_i = s_j$.
Then $\myvec u \to_\RR^* \myvec w$ and the
imbalance of $\myvec w$ is at most $\iota$. Hence
$C[\myvec u] \rrr[\leq\iota]{} {\cdot} \RRR{*} C[\myvec w]$ by
Lemma~\ref{lem-tall-short}.
We also have $B[\myvec v] \to_\RR^* C[\myvec w]$ with no rewrite step
affecting a tall alien and thus $B[\myvec v] \RRR{*} C[\myvec w]$
by Lemma~\ref{lem-shallow-step}.
\end{proof}

\begin{example}
We revisit Example~\ref{ex-vO08b}.
Let $s = \m f(\m f(\m G(\m a), \m G(\m a)),\m I)$.
The base context of $s$ is $\m f(\m f(\square,\square),\m I)$.
Then
$t = \m f(\m f(\m I, \m H(\m a)), \m I) \lll[2]{} s
\RRR{} \m f(\m G(\m a),\m K) = u$. The base context of
$t$ is $\m f(\m f(\m I, \square), \m I)$ and we have
$t \rrr[1]{} \m f(\m f(\m I,\m K),\m I)
   \RRR{} \m f(\m f(\m K,\m K),\m K)
   \RRR{} \m f(\m K, \m K) = v$,
whereas the base context of $u$ is $\m f(\square,\m K)$ and
$u \rrr[1]{} v$.
\end{example}

\begin{lemma}[(Main Lemma)]
\label{lem-main}
If $\RRR{}$ is locally decreasing then $\RR$ is confluent on native terms.
\end{lemma}

\begin{proof}
Every rewrite step $s \to_\RR t$ can be written as $s \RRR{} t$ by
Lemma~\ref{lem-short-step} or $s \rrr{} t$ if the rewrite rule
is applied to a tall alien of $s$. Hence
${{\to_\RR} \subseteq {\rrr{}} \cup {\RRR{}} \subseteq {\to_\RR^*}}$ and
thus the claim follows from the confluence of
${\rrr{}} \cup {\RRR{}}$. The latter is a consequence of 
Theorem~\ref{thm-dd} in connection with the assumption and 
Lemmata~\ref{lem-peak-tt} and~\ref{lem-peak-ts}.
\end{proof}

The various versions of the main theorem will follow from
Lemma~\ref{lem-main}.

\subsection{Local Decreasingness of Short Steps}
\label{sec-dd-short}

In this section we study conditions to make short steps locally decreasing.
The following result allows to represent a native term $s$ by a foreign
term $s'$ and a substitution $\pi$ such that $s = s'\pi$. This will be
the key for joining the peak originating from~$s$ by the confluence
assumption of~$s'$.

\begin{lemma}[(peak analysis)]
\label{lem-peak-analysis}
For a local peak $t \LLL{} s \RRR{} u$ there are foreign terms
$s'$, $t'$, $u'$, $v'$ and substitutions $\pi$, $\pi_\square$ such that
\begin{enumerate}
\item
$\pi$ is a bijection with $\dom(\pi) \cap \Var(s) = \varnothing$,
\item
$s'\pi = s$, $t'\pi = t$, $u'\pi = u$, $s'\pi_\square$ is the
base context of $s$, and $t' \pi_\square$ and $u' \pi_\square$ are
shallow contexts of $t$ and $u$, and
\item
$v' \lud{\RR}{*}\from t' \lud{\RR}{*}\from s' \to_\RR^* u' \to_\RR^* v'$
and $t \to_\RR^* v \lud{\RR}{*}\from u$ with $v = v' \pi$.
\end{enumerate}
\end{lemma}

\begin{proof}
Let $s = B[\myvec s]$ be the decomposition of $s$ into base context
and base sequence, and recall that base contexts are shallow. According
to the definition of $\RRR{}$ there are
rewrite sequences $B \to_\RR^* C_t$, $B \to_\RR^* C_u$
mirroring $s \to_\RR^* t$, $s \to_\RR^* u$, respectively.
Using Lemma~\ref{lem-shallow-closed} repeatedly, we find that
$C_t$ and $C_u$ are shallow contexts.
Let $\pi$ be a bijection between the
tall aliens of $s$ and fresh variables,
and define $s' = B[\pi^{-1}(\myvec s)]$.
We have $\myvec s \propto \pi^{-1}(\myvec s)$ and therefore repeated
application of Lemma~\ref{lem-balanced-rewrite} yields rewrite sequences
$s' \to_\RR^* t'$ and $s' \to_\RR^* u'$
mirroring $s' \pi = s \to_\RR^* t = t' \pi$ and
$s' \pi = s \to_\RR^* u = u'\pi$.
Since $s'$ is a foreign term and therefore
confluent, $t'$ and $u'$ have a
common reduct: $t' \to_\RR^* v' \lud{\RR}{*}\from u'$. By applying
$\pi$ to this valley we obtain $t \to_\RR^* v \lud{\RR}{*}\from u$.
Note that $s'\pi_\square = B$, $t'\pi_\square = C_t$ and
$u'\pi_\square = C_u$ are shallow contexts as claimed.
\end{proof}

\begin{example}
Consider the layer system $\bL$ given by
\begin{align*}
\bL_0 &= \{
v, \m a, \m b, \m f(v), \m g(v), \m g(\m b) \mid v \in \VC
\} \\
\bL &= \bL_0 \cup \{ \m h(C,C',C'') \mid C,C',C'' \in \bL_0 \}
\end{align*}
which weakly layers the TRS
$\RR = \{ \m h(x,y,z) \to \m h(y,x,z), \m f(x) \to \m g(x),
\m a \to \m b \}$.
Assume that $r = 1$
and let $s = \m h(\m a, \m f(\m a), \m f(\m b))$. The base context
of $s$ is $\m h(\m a, \m f(\square), \m f(\square))$. There is a peak
of short steps
\[t = \m h(\m b, \m g(\m a), \m f(\m b)) \LLL{} s \RRR{}
\m h(\m f(\m a), \m a, \m g(\m b)) = u
\]
From Lemma~\ref{lem-peak-analysis},
we may obtain $\pi = \{ \m a / x, \m b / y \}$,
$s' = \m h(\m a,\m f(x),\m g(y))$,
$t' = \m h(\m b, \m g(x), \m f(y))$,
$u' = \m h(\m f(x), \m a, \m g(y))$,
and $v' = \m h(\m g(x), \m b, \m g(y))$.
Note that $t' \pi_\square = \m h(\m b, \m g(\square), \m f(\square))$
is the base context of $t$ but
$u' \pi_\square = \m h(\m f(\square), \m a, \m g(\square))$
does not equal $\m h(\m f(\square), \m a, \m g(\m b))$, the
base context of $u$.
\end{example}

\begin{lemma}
\label{lem-ss-ll}
If $\RR$ is left-linear then $\RRR{}$ is locally decreasing.
\end{lemma}

\begin{proof}
Consider a local peak $t \LLL[s_0]{} s \RRR[s_1]{} u$. First we apply
Lemma~\ref{lem-peak-analysis}. Let $t''$ be a linearization of $t'$, which
we obtain by replacing each variable in $t'$ by a fresh variable. Because
$\RR$ is left-linear, $t' \to_\RR^* v'$
can be mirrored as $t'' \to_\RR^* v''$. Let $B_t$ be the base context of
$t$ and $C_t = t'\pi_\square$. We have $C_t \Cleq B_t$ by
Lemma~\ref{lem-base-increase}, which implies $t'' \matches B_t$
and thus $B_t = t''\sigma$ for some substitution $\sigma$.
We have
$B_t \to_\RR^* v''\sigma$. Together with $t \to_\RR^* v$,
which mirrors $B_t \to_\RR^* v''\sigma$,
we obtain $t \RRR{} v$. This step can be labeled
with $s_1$ because $s_1 \to_\RR^* s \to_\RR^* t$.
By symmetry we obtain $u \RRR[s_0]{} v$ and hence $\RRR[]{}$ is locally
decreasing.
\end{proof}

Next we deal with bounded duplicating TRSs.
In order to exploit relative termination,
we insert $\Diamond$ symbols in front of tall aliens as follows.

\begin{definition}
Let $s$ be a native term
with base context $B$ and base sequence~$\myvec s$.
Then
$s^\Diamond = B[\Diamond(\myvec s)]$ where $\Diamond(\myvec s)$
denotes the result of 
replacing each element $u$ of $\myvec s$ by $\Diamond(u)$.
\end{definition}

\begin{lemma}
\label{lem-diamond-step}
If $s \to_\RR t$ then
$s^\Diamond \to_\RR {\cdot} \to_{\Diamond(x) \to x}^* t^\Diamond$.
\end{lemma}

\begin{proof}
Let $s \to_{p,\ell \to r} t$ and let $B$ be the base context of $s$.
If $p \in \Pos_\FF(B)$ then by Lemma~\ref{lem-balanced-rewrite}
we obtain a term $t'$ and a context $C$ such that
$s^\Diamond \to_{p,\ell \to r} t'$ and $B \to_{p,\ell \to r} C$.
Decomposing $t$ as $t = C[\myvec t]$ we find that
$t' = C[\Diamond(\myvec t)]$.
If $p \notin \Pos_\FF(B)$, then the rewrite step is within a tall alien
of $s$.
Hence letting $C = B$ and decomposing $t$ as $C[\myvec t]$,
we find that $s^\Diamond = C[\Diamond(\myvec s)] \to_\RR
C[\Diamond(\myvec t)]$.
In either case, Lemma~\ref{lem-fusion} (with $L$ equal to the max-top of
$C$)
shows that the tall aliens of $t$ are a subsequence of $\myvec t$,
and therefore $C[\Diamond(\myvec t)] \to_{\Diamond(x) \to x}^* t^\Diamond$,
using that $\Diamond(t_i) \to_{\Diamond(x) \to x} t_i$
for those $t_i$ that are not tall aliens.
\end{proof}

\begin{lemma}
\label{lem-ss-bd}
If $\RR$ is bounded duplicating then $\RRR{}$ is locally decreasing.
\end{lemma}

\begin{proof}
Since $\RR$ is bounded duplicating, we may assume a fresh function
symbol $\Diamond$ such that $\{ \Diamond(x) \to x \} / \RR$ is
terminating. In order to compare the labels
we define a well-founded order on 
terms by $s_0 \succ s_1$ if
$s_0^\Diamond \to_{\{ \Diamond(x) \to x \} / \RR}^+ s_1^\Diamond$.
Consider a local peak $t \LLL[s_0]{} s \RRR[s_1]{} u$ which we first
subject to
Lemma~\ref{lem-peak-analysis}. We analyze the sequence $t \to_\RR^* v$
resulting from the peak analysis by distinguishing two cases.
\begin{enumerate}
\item
If $t'\pi_\square$ is the base context of $t$ then the rewrite sequence
$t'\pi_\square \to_\RR^* v'\pi_\square$ mirrors $t \to_\RR^* v$.
Hence we obtain
$t \RRR[s_1]{} v$, noting that the label $s_1$ satisfies
$s_1 \to_\RR^* s \to_\RR^* t$.
\item
If $t'\pi_\square$ is not the base context
then like in the proof of Lemma~\ref{lem-diamond-step},
we can decompose $t$ as $t = t'\pi_\square[\myvec t']$ in order to obtain
$s^\Diamond \to_\RR^* t'\pi_\square[\Diamond(\myvec t')]$.
Since $t'\pi_\square$ is not the base context,
the tall aliens of $t$ are a proper subsequence of $\myvec t'$
and therefore,
$t'\pi_\square[\Diamond(\myvec t')] \to_{\Diamond(x) \to x}^+ t^\Diamond$.
We also have $s_1 \to_\RR^* s$, which implies
$s_1^\Diamond \to_{\RR \cup \{ \Diamond(x) \to x \}}^* s$
by Lemma~\ref{lem-diamond-step}.
As a consequence,
$s_1^\Diamond \to_{\RR / \{ \Diamond(x) \to x \}}^+ t^\Diamond$
and $s_1 \succ t$ follow.
By repeated application of Lemma~\ref{lem-short-step} we obtain
$t \RRR[t]{*} v$ and thus $t \RRR[\curlyvee s_1]{*} v$.
\end{enumerate}
The analogous analysis of $u \to_\RR^* v$ yields $u \RRR[s_0]{} v$ or
$u \RRR[\curlyvee s_0]{*} v$ and hence $\RRR[]{}$ is locally decreasing.
\end{proof}

Finally, we prepare for the main result about layered TRSs, where
condition \cconsistent of Definition~\ref{def-laysys} is crucial.

\begin{lemma}
\label{lem-consistent}
Let $\RR$ be a layered TRS and $t \to_{p,\ell \to r} t'$
for native terms $t$ and $t'$
If $p \in \Pos_\FF(B)$ for the base context $B$ of $t$ then
either $B \to_{p,\ell \to r} B'$ for the base context $B'$ of $t'$
or $t'$ is its own base context.
\end{lemma}

\begin{proof}
Let $M$ and $M'$ be the max-tops of $t$ and $t'$. We distinguish two
cases.
\begin{enumerate}
\item
If $p \in \Pos_\FF(M)$ then by
\cconsistent either $M \to_{p,\ell \to r} \square$ or
$M \to_{p,\ell \to r} M'$. In the former case $t'$ equals an
alien of $t$. Since the rank of $t'$ is at most $r$, $t'$ is its
own base context. So assume $M \to_{p,\ell \to r} M'$.
By Lemma~\ref{lem-context-instance} there exist a term $m$
and a substitution $\sigma$
such that $m \to_{p,\ell \to r} m'$ for some $m'$ (since 
$\ell \matches m|_p$), $t = m\sigma$, and $M = m\sigma_\square$.
Define a substitution $\tau$ as follows:
\[
\tau(x) = \begin{cases}
\square & \text{if $x \in \dom(\sigma_\square)$ and
$\sigma(x)$ is a tall alien of $t$} \\
\sigma(x) & \text{otherwise}
\end{cases}
\]
We have $B = m\tau$ by construction of $\tau$. Let $B' = m'\tau$.
Clearly $B \to_{p,\ell \to r} B'$.
By comparing $m'\tau$ to $M' = m'\sigma_\square$, we see that $B'$ is
the base context of $t'$.
\item
If $p \notin \Pos_\FF(M)$ then a short alien of $t$ is rewritten.
By letting $B$ and $\myvec t$ be the base context and base sequence
of $t$, by Lemma~\ref{lem-context-rewrite} we obtain a rewrite step
$t = B[\myvec t] \to_{p,\ell \to r} B'[\myvec t'] = t'$
with $\myvec t' = \myvec t$ because $p$ is parallel to the
hole positions of $B$.
We claim that $B'$ is the base context of $t'$.
Suppose to the contrary that some $t_i$ is not a tall alien
of $t'$. Let $q$ be its position in $t$,
which is also its position in $t'$. Since $q \in \Pos_{\square}(M)$
and $q \notin \Pos_{\square}(M')$, $M \sqsubset M[M'|_q]_q$. Hence
$M[M'|_q]_q \in \bL$ by \cstepfusion and thus
$M[M'|_q]_q \sqsubseteq t$, contradicting
the fact that $M$ is a max-top of $t$.
\qed
\end{enumerate}
\end{proof}

The following example shows that \cstepfusion is essential for
Lemma~\ref{lem-consistent}.

\begin{example}
\label{ex-stepfusion-essential}
Recall Figure~\ref{fig-problems} and the underlying layer system $\bL$, 
which satisfies \crewrite and \cconsistent. However \cstepfusion is
violated, e.g., we have
$L = \m{k}(\square,\square) \in \bL$ and
$N = \m{k}(\m{h}(\square),\m{h}(\square)) \in \bL$ but
$L[N|_2]_2 = \m{k}(\square,\m{h}(\square)) \notin \bL$.
Consider the term $t = \m{k}(\m{f}(\m{a}),\m{h}(\m{a}))$ of rank $3$.
Its base context is $B = \m{k}(\m{f}(\m{a}),\square)$.
We have $t \to \m{k}(\m{h}(\m{a}),\m{h}(\m{a})) =: t'$.
The base context of $t'$ is $\m{k}(\m{h}(\square),\m{h}(\square)) =: B'$ 
but $B \not\to_\RR B'$.
\end{example}

\begin{lemma}
\label{lem-ss-con}
If $\RR$ is layered then $\RRR{}$ is locally decreasing.
\end{lemma}

\begin{proof}
Consider a local peak $t \LLL[s_0]{} s \RRR[s_1]{} u$. First we analyze
the peak by Lemma~\ref{lem-peak-analysis}.
The rewrite sequence $t'\pi_\square \to_\RR^* v'\pi_\square$ mirrors
$t = t'\pi \to_\RR^* v'\pi = v$.
We find by repeated application
of Lemma~\ref{lem-consistent} that the base context $B_t$ of $t$ equals
$t'\pi_\square$ or $t$.
In both cases, we have $t \RRR[s_1]{} v$, noting that $t \to_\RR^* v$
mirrors itself, and that $s_1 \to_\RR^* t$.
We obtain $u \RRR[s_0]{} v$ in the same way and hence 
$\RRR[]{}$ is locally decreasing.
\end{proof}

\subsection{Proof of Main Theorems}
\label{sec-dd-proof}

Because the proofs are similar, we prove all main results in one go.

\begin{proof}[of Theorems~\ref{thm-main-ll}, \ref{thm-main-bd},
and \ref{thm-main-con}]
By assumption 
the TRS~$\RR$ is weakly layered and confluent on terms of rank~one.
We have to show that
\begin{itemize}
\item[~~-]
if $\RR$ is left-linear then $\RR$ is confluent
(Theorem~\ref{thm-main-ll}),
\item[~~-]
if $\RR$ is bounded duplicating then $\RR$ is confluent
(Theorem~\ref{thm-main-bd}), and
\item[~~-]
if $\RR$ is layered then $\RR$ is confluent (Theorem~\ref{thm-main-con}).
\end{itemize}
We show confluence of 
all terms by induction on the rank~$r$ of a term. In the base case
we consider terms of rank one, which are confluent by assumption.
Assume as induction hypothesis that confluence of terms of rank $r$
or less has been established. We consider terms of rank $r+1$,
to which the analysis of Sections~\ref{sec-setup}--\ref{sec-dd-short}
applies. By Lemma~\ref{lem-main} in conjunction
with Lemma~\ref{lem-ss-ll} (for weakly layered left-linear~$\RR$),
Lemma~\ref{lem-ss-bd} (for weakly layered bounded duplicating~$\RR$),
or Lemma~\ref{lem-ss-con} (for layered~$\RR$),
we obtain confluence of $\RR$ on terms of rank up to $r+1$,
completing the induction step.
\end{proof}

\section{Applications}
\label{sec-applications}

In this section the abstract confluence results via layer systems are
instantiated by concrete applications. Section~\ref{app-modularity}
treats the plain modularity case~\cite{T87} and
Section~\ref{app-layer-preservation}
covers layer-preservation~\cite{O94}. The result for quasi-ground
systems~\cite{KST95}
is less known but also fits our framework, as outlined
in Section~\ref{app-quasi-ground}. Currying~\cite{K95} is the topic of
Section~\ref{app-currying}, before many-sorted persistence~\cite{AT97} is
discussed in Section~\ref{app-many-sorted}. 

For the results in this section the reverse directions also hold. We do
not give the (easy) proofs since they do not require layer systems.

In Sections~\ref{app-modularity},
\ref{app-layer-preservation}, and
\ref{app-quasi-ground} we deal with two TRSs $\RR_1$ and $\RR_2$ that
are defined over the respective signatures $\FF_1$ and $\FF_2$.
We let $\RR = \RR_1 \cup \RR_2$ and $\FF = \FF_1 \cup \FF_2$.

\subsection{Modularity}
\label{app-modularity}

We recall the classical modularity result for confluence~\cite{T87}.

\begin{theorem}
\label{thm-modularity}
Suppose $\FF_1 \cap \FF_2 = \varnothing$.
If $\RR_1$ and $\RR_2$ are confluent then $\RR$ is confluent.
\end{theorem}
\begin{proof}
Define
\[
\bL := \CC(\FF_1,\VV) \cup \CC(\FF_2,\VV)
\]
We show that $\RR$ is layered.
Since $\VV \subseteq \bL$ and $f(\square,\dots,\square) \in \bL$ for all
function symbols $f \in \FF_1 \cup \FF_2$, every term in
$\TT(\FF,\VV)$ has a non-empty top. Hence condition \ctop holds.
Also condition \cvar holds because $\bL$ is closed under
the operation of interchanging variables and holes.
For condition \cpartial we observe that if
${L \in \CC(\FF_i,\VV)}$, $p \in \Pos_\FF(L)$, and
$N \in \bL$ such that $L|_p \merge N$ is defined then
$\rt(L|_p) \in \FF_i$ and thus
$N \in \CC(\FF_i,\VV)$. Consequently,
$L[L|_p \merge N]_p \in \CC(\FF_i,\VV) \subseteq \bL$.
Since each rule is over a single signature, and layers are closed
under rewriting, condition \crewrite follows easily.
For condition \cconsistent we consider a 
term~$s$ with max-top~$M$, $p \in \Pos_\FF(M)$, and
rewrite step $s \to_{p,\ell \to r} t$ which is mirrored
by $M \to_{p,\ell \to r} L$.
Suppose $M \in \CC(\FF_i,\VV)$. We have $L \in \CC(\FF_i,\VV)$.
The case $L = \square$ is obtained when $t$ is an alien of $s$,
which is only possible if the rule $\ell \to r$ is collapsing.
Otherwise $L$ is the max-top of $t$
since the root symbols of aliens of $s$ belong to $\FF_{3-i}$ and
hence cannot fuse with $L$ to form a larger top.
Finally, condition \cstepfusion holds because if
$N \in \CC(\FF_i,\VV)$ then
$L \Cleq N$ implies $L \in \CC(\FF_i,\VV)$
and thus also $L[N|_p]_p$ belongs to
$\CC(\FF_i,\VV)$.

According to Theorem~\ref{thm-main-con}, $\RR$ is confluent
if we show that $\RR$ is confluent on terms of rank one. The latter follows
from the fact that rewriting does not increase the rank of a term
(Lemma~\ref{lem-rank-decrease}) together with the observation
that non-variable terms of rank one belong to either
$\TT(\FF_1,\VV)$ or $\TT(\FF_2,\VV)$ and only rewrite rules of $\RR_i$
apply to terms in $\TT(\FF_i,\VV)$, in connection with the confluence
assumptions of $\RR_1$ and $\RR_2$.
\end{proof}

\subsection{Layer-Preservation}
\label{app-layer-preservation}

Layer-preserving TRSs are a special class of TRSs with
shared function symbols
for which confluence is modular as shown
in \cite{O94}. In this section, we reprove this result using
layer systems.
Let $\TT_X(\FF,\VV)$ denote the set of terms with root symbol from $X$.
Let $\BB := \FF_1 \cap \FF_2$, $\DD_1 := \FF_1 \setminus \FF_2$ and
$\DD_2 := \FF_2 \setminus \FF_1$.
The result on layer preservation can be stated as follows.

\begin{theorem}
\label{thm-layer-preservation}
Let 
$\RR_1 \subseteq \TT(\BB,\VV)^2 \cup \TT_{\DD_1}(\FF_1,\VV)^2$,
$\RR_2 \subseteq \TT(\BB,\VV)^2 \cup \TT_{\DD_2}(\FF_2,\VV)^2$, and
$\RR_1 \cap \TT(\BB,\VV)^2 = \RR_2 \cap \TT(\BB,\VV)^2$.
If $\RR_1$ and $\RR_2$ are confluent then $\RR$ is confluent.
\end{theorem}

\begin{proof}
We define
\[
\bL := \CC(\BB,\VV) \cup
\TT_{\DD_1}(\FF_1 \cup \{ \square \},\VV) \cup
\TT_{\DD_2}(\FF_2 \cup \{ \square \},\VV)
\]
It is easy to verify that $\bL$ layers
$\RR := \RR_1 \cup \RR_2$, much like in the modularity case.
In particular, $\bL$ is closed under rewriting.
Consider a term $s$ of rank one and a peak
$t \lud{\RR}{*}\from s \rud{\RR}{*}\to u$.
Let $i \in \{ 1, 2 \}$ be such that
$s \in \TT(\FF_i,\VV)$. The only rules of $\RR_{3-i}$ that can be
used in the peak come from $\TT(\BB,\VV)^2$ and hence
also appear in
$\RR_i$. Since $\RR_i$ is confluent on $\TT(\FF_i,\VV)$
we obtain joinability of $t$ and $u$ in $\RR_i$ and thus also in~$\RR$.
Hence $\RR$ is confluent on terms of rank one and we conclude by
Theorem~\ref{thm-main-con}.
\end{proof}

Toyama's modularity result has been adapted in~\cite{O94a}
to constructor-sharing combinations in which the participating TRSs may
share constructor symbols under the additional condition that
neither collapsing nor constructor-lifting rules are present.
This result is subsumed by Theorem~\ref{thm-layer-preservation},
cf.~\cite[p.~249]{O02}.
Still, layer-preservation and modularity are incomparable
(since layer-preservation places collapsing rules in both systems).

\subsection{Quasi-Ground Systems}
\label{app-quasi-ground}

We show modularity of quasi-ground TRSs \cite[Theorem 1]{KST95} using
layer systems.

\begin{definition}
\label{def-quasi-ground}
We call a context $C$ \emph{quasi-ground} if for all
$p \in \Pos(C)$ with $\rt(C|_p) \in \FF_1 \cap \FF_2$,
$C|_p$ is ground over $\FF$, i.e., $C|_p \in \TT(\FF)$.
\end{definition}

\begin{theorem}
Suppose
$\rt(\ell) \notin \FF_1 \cap \FF_2$ and $\ell$ and $r$ are quasi-ground,
for all $\ell \to r \in \RR$.
If $\RR_1$ and $\RR_2$ are confluent then $\RR$ is confluent.
\end{theorem}

\begin{proof}
We define a layer system $\bL := \bL_1 \cup \bL_2 \cup \bL_c$ with
\begin{align*}
\bL_i &= \{ C \in \CC(\FF_i,\VV) \mid
\text{$C$ is quasi-ground} \} \quad \text{for $i = 1, 2$} \\
\bL_c &= \{ f(\seq{v}) \mid
\text{$f \in \FF_1 \cap \FF_2$ and $v_i \in \VC$ for $1\leq i\leq n$}\}
\end{align*}
We readily check that \ctop and \cvar are satisfied. For
\cpartial, $\bL_1$, $\bL_2$ and $\bL_c$ are individually
closed under merging at function positions. Fix $i \in \{ 1, 2 \}$.
If we merge $L \in \bL_i$ with $N \in \bL_{3-i} \cup \bL_c$
at $p \in \Pos_\FF(L)$ then either $N = \square$ and
$L[L|_p \merge N] = L \in \bL_i$, or
$\rt(L|_p) \in \FF_1 \cap \FF_2$, which implies
$L|_p \in \TT(\FF)$ and hence
$L[L|_p \merge N]_p = L[L|_p]_p = L \in \bL_i$.
Note that $L \in \bL_c$ can be merged with $N \in \bL_i$ only 
at position $p = \epsilon$. If $N = \square$ then
$L \merge \square = L \in \bL_c$
and otherwise $L \merge N = N \in \bL_i$.
For \crewrite we
let $M$ be the max-top of $s$, $p \in \Pos_\FF(M)$,
and consider a rewrite step $s \to_{p,\ell \to r} t$.
We assume without
loss of generality that $\ell \to r \in \RR_1$. Hence
$M \in \bL_1$ because $\rt(\ell) \in \FF_1 \setminus \FF_2$.
Note that $\bL_1$ is closed under taking subterms and that for any
substitution $\tau\colon \VV \to \bL_1$ we have $\ell \tau \in \bL_1$.
Let $\sigma$ be a substitution such that $s|_p = \ell\sigma$ and
let $\tau$ be the substitution that maps each
variable $x \in \Var(\ell)$ to the $\bL_1$-max-top of $\sigma(x)$.
We have $M = M[\ell\tau]_p$ and thus $M \to_{p,\ell \to r} L$ with
$L = M[r\tau]_p \in \bL_1$. For \cconsistent it is easy to see that $L$
is the $\bL_1$-max-top of $t$. Suppose $L \neq \square$. We
claim that $L$ is the max-top (with respect to $\bL$) of $t$. This follows
from the observation that if there is a top of $t$ that comes from
$\bL_2$ or $\bL_c$ then $\rt(L) \in \FF_1 \cap \FF_2$ and thus
$L \in \TT(\FF)$, which cannot be made larger.
Condition \cstepfusion follows as in the proof of
Theorem~\ref{thm-modularity}.

Now let $\RR_1$ and $\RR_2$ be confluent. We show that $\RR$ is confluent
on terms of rank one. Consider a term of rank one. Note that rules from
$\RR_1$
only apply to elements of $\bL_1$. Furthermore, $\bL_1$ is closed under
rewriting by $\RR_1$. Likewise, rules from $\RR_2$
only apply to elements of $\bL_2$, which is closed under rewriting
by $\RR_2$. We conclude that $\RR$ is confluent on terms of rank one and by
Theorem~\ref{thm-main-con} this implies that $\RR$ is confluent.
\end{proof}

\subsection{Currying}
\label{app-currying}

Currying is a transformation of TRSs such that the resulting TRS has
only one non-constant function symbol $\m{Ap}$ that represents
partial applications. It is useful in the construction of
polynomial-time procedures for deciding properties of TRSs, e.g.,
\cite{CGN01}. \cite{K95}~proved that confluence is
preserved by currying.

\begin{definition}
Given a TRS $\RR$ over a signature $\FF$, let 
$\FF_\CC = \{ \m{Ap} \} \cup \{ f_0 \mid f \in \FF \}$ where $\m{Ap}$ is
a fresh binary
function symbol and all function symbols in $\FF$ become constants.
The \emph{curried version} $\Cu(\RR)$ of $\RR$ is the TRS over the
signature $\FF_\CC$ with rules
$\{ \Cu(\ell) \to \Cu(r) \mid \ell \to r \in \RR \}$. Here
$\Cu(t) = t$ if $t$ is a variable or a constant and
$\Cu(f(\seq{t})) =
\m{Ap}( \cdots \m{Ap}(f_0,\Cu(t_1)) \cdots , \Cu(t_n))$
(with $n$ occurrences of $\m{Ap}$).
Let $\FF_\UU = \{ \m{Ap} \} \cup \{ f_i \mid \text{$f \in \FF$ and
$0 \leqslant i \leqslant \ari(f)$} \}$, where 
each $f_i$ has arity $i$ and $f_{\ari(f)}$ is identified with $f$.
The \emph{partial parametrization} $\PP(\RR)$ of $\RR$ is the TRS
$\RR \cup \UU$ over the signature $\FF_\UU$,
where $\UU$ consists of all \emph{uncurrying} rules:
\[
\m{Ap}(f_i(\seq[i]{x}),x_{i+1}) \to f_{i+1}(\seq[i+1]{x})
\]
for all $f \in \FF$ and $0 \leqslant i < \ari(f)$.
\end{definition}

The next example familiarizes the reader with the above concepts.

\begin{example}
\label{ex-currying}
For the TRS $\RR = \{ \m{f}(x,x) \to \m{f}(\m{a},\m{b}) \}$ we have
\begin{align*}
\Cu(\RR) &= \{ \m{Ap}(\m{Ap}(\m{f}_0,x),x) \to
\m{Ap}(\m{Ap}(\m{f}_0,\m{a}),\m{b}) \} \\
\UU      &= \{ \m{Ap}(\m{f}_0,x) \to \m{f}_1(x),
               \m{Ap}(\m{f}_1(x),y) \to \m{f}(x,y) \} \\
\PP(\RR) &= \RR \cup \UU
\end{align*}
Note that for a term $s = \m{Ap}(\m{Ap}(\m{Ap}(\m{f}_0,x),x),x)$
we have
\begin{align*}
s &\to_{\Cu(\RR)} \m{Ap}(\m{Ap}(\m{Ap}(\m{f}_0,\m{a}),\m{b}),x)
\intertext{and}
s &\to_\UU \m{Ap}(\m{Ap}(\m{f}_1(x),x),x) 
   \to_\UU \m{Ap}(\m{f}(x,x),x)
   \to_{\RR} \m{Ap}(\m{f}(\m{a},\m{b}),x)
\end{align*}
so the partial parametrization is closely related to currying.
\end{example}

Note that $\UU$ is both terminating and orthogonal, hence confluent. 
By $\rnf[\UU]{s}$ we denote the unique $\UU$-normal form of a term~$s$.

\begin{lemma}[{\cite[Proposition~3.1]{K95}}]
\label{lem-cur-pp-cur}
Let $\RR$ be a TRS.
If $\PP(\RR)$ is confluent then $\Cu(\RR)$ is confluent.
\qed
\end{lemma}

\begin{theorem}[{\cite[Theorem~5.2]{K95}}]
Let $\RR$ be a TRS.
If $\RR$ is confluent then $\Cu(\RR)$ is confluent.
\end{theorem}
\begin{proof}
According to Lemma~\ref{lem-cur-pp-cur} it suffices to
show that $\PP(\RR)$ is confluent.
To this end, we let $\bL := \bL_1 \cup \bL_2$, where
$\bL_1$ is the smallest extension of $\VC$ such that
\[
\m{Ap}(\cdots \m{Ap}(f_m(\seq[m]{s}),s_{m+1}) \cdots, s_n) \in \bL_1
\]
for all $f_m \in \FF_\UU \setminus \{ \m{Ap} \}$, $\seq{s} \in \bL_1$, 
with $n$ less than or equal to
the arity of $f$ in the original TRS~$\RR$, and
\[
\bL_2 =
\{ \m{Ap}(v,t) \mid \text{$v \in \VC$ and $t \in \bL_1$} \}
\]
It is not difficult to see that $\bL_1$ consists of those
contexts in $\CC(\FF_\UU,\VV)$
whose $\UU$-normal form contains no occurrences of $\m{Ap}$.
See Figure~\ref{fig-app} for some layered terms.

\begin{figure}[t]
\small
\subfloat[\label{fig-currying:1}]{%
\input{samples/app0.tikz}
}
\hfill
\subfloat[\label{fig-currying:2}]{%
\input{samples/app1.tikz}
}
\hfill
\subfloat[\label{fig-currying:3}]{%
\input{samples/app2.tikz}
}
\hfill
\subfloat[\label{fig-currying:4}]{%
\input{samples/app3.tikz}
}
\caption{Layering terms in $\PP(\RR)$ for the TRS $\RR$ in
Example~\ref{ex-currying}.}
\label{fig-app}
\end{figure}
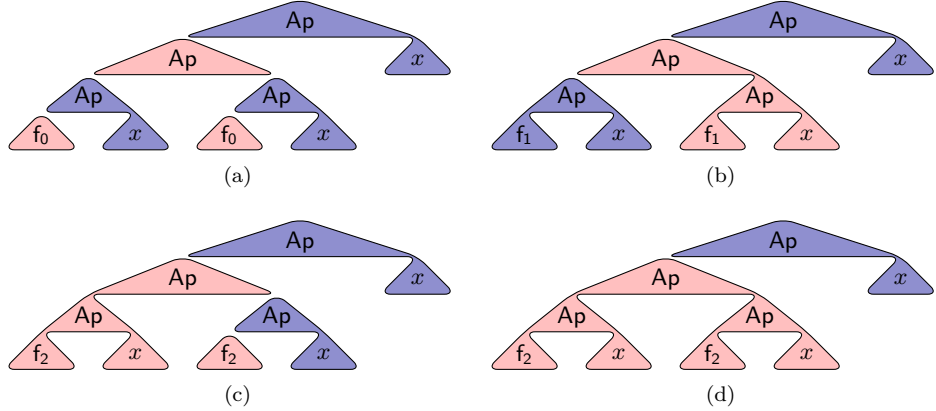

We claim that $\PP(\RR)$ is layered. Conditions \ctop and \cvar are
trivial and conditions \cpartial and \cstepfusion are easily shown by
induction on the definition of $\bL_1$. 
The interesting case for \cpartial is when $L \in \bL_1$. Since
merging cannot create new
$\m{Ap}$ symbols above any $f_m$, the result is in
$\bL_1$, whenever defined.
For \crewrite and \cconsistent, we
let $M$ be the max-top of $s$, $p \in \Pos_\FF(M)$, and
consider a rewrite step $s \to_{p,\ell \to r} t$
with $\ell \to r \in \PP(\RR)$
Because $\bL$ is closed under taking subterms,
$M|_p$ is a top of $s|_p$. It is the max-top because otherwise
we could merge the max-top of $s|_p$ with $M$ at position $p$ and obtain
a larger top of $s$.
Note that $\ell_\square \in \bL_1$
(recall that $\ell_\square$ is obtained by replacing all variables in
$\ell$ by $\square$)
We have $\ell_\square \Cleq s|_p$ and 
therefore $\ell_\square \Cleq M|_p$. As a matter of fact,
$M|_p$ is obtained from $\ell_\square$ by replacing
each hole at position $q$ by the max-top (in $\bL_1$) of $s|_{pq}$.
Because equal subterms have equal max-tops, $s \matches M|_p$ and
hence there is a rewrite step $M \to_{p,\ell \to r} L$.
We have $L \in \bL_1$ because $\bL_1$ is closed under rewriting
by $\PP(\RR)$. Furthermore,
the max-tops of the aliens of $s$ do not belong to $\bL_1$
and therefore the aliens of $s$ are still aliens of $L$,
unless $L = \square$.
It follows that both \crewrite and \cconsistent hold.

To show confluence of $\PP(\RR)$ on terms of rank one, first note that
elements of $\bL_2$ allow no root steps and therefore it suffices to
show confluence on terms in $\bL_1$. It is easy to see that
$s \to_{\RR\,\cup\,\UU} t$ implies $\rnf[\UU]{s} \to_\RR^= \rnf[\UU]{t}$.
Hence, for a peak $t \lud{\RR\,\cup\,\UU}{*}\from s
\to_{\RR\,\cup\,\UU}^* u$
there is a corresponding peak $\rnf[\UU]{t} \lud{\RR}{*}\from
\rnf[\UU]{s} \to_\RR^* \rnf[\UU]{u}$, which is joinable by
the confluence of $\RR$. Hence $t$ and $u$ are joinable in $\PP(\RR)$.
We conclude by Theorem~\ref{thm-main-con}.
\end{proof}

\subsection{Many-sorted Persistence}
\label{app-many-sorted}

In this subsection, we prove persistence of confluence \cite{AT96}.
We begin by recalling many-sorted terms and rewriting.

\begin{definition}
\label{def-sorted-signature}
Let $S$ be a set of \emph{sorts}.
A \emph{sort attachment $\SS$ associates
with each function symbol $f \in \FF$ of arity $n$
a type $f : \alpha_1 \times \cdots \times \alpha_n \to \alpha$
with $\alpha_i, \alpha \in S$ for $1\leq i\leq n$, and
with each variable $x \in \VV$ a sort from $S$.}
Let $\VV_\alpha$ denote the set of variables of sort $\alpha$.
We assume that each $\VV_\alpha$ is countably infinite.
\end{definition}

Note that $\VV_\alpha \cap \VV_\beta = \varnothing$ for all
$\alpha,\beta \in S$ whenever $\alpha \neq \beta$.

\begin{definition}
\label{def-typed-terms}
Let $\SS$ be a sort attachment.
We define terms of sort $\alpha$ inductively by
$\TT_\alpha(\FF,\VV) = \VV_\alpha \cup \{
f(\seq{t}) \mid
f : \alpha_1 \times \cdots \times \alpha_n \to \alpha
\text{ and } t_i \in \TT_{\alpha_i}(\FF,\VV)
\text{ for } 1 \leq i \leq n
\}$.
The set of many-sorted terms is defined as
$\TT_\SS(\FF,\VV) = \bigcup_{\alpha \in S} \TT_\alpha(\FF,\VV)$.
\end{definition}

\begin{definition}
A TRS $\RR$ is \emph{compatible} with a sort attachment $\SS$ if 
for each rule $\ell \to r \in \RR$, there is a sort $\alpha \in S$
with $\ell,r \in \TT_\alpha(\FF,\VV)$.
\end{definition}

\begin{remark}
\label{rem-sort-closed}
If a TRS $\RR$ is compatible with a sort attachment $\SS$ then
$\TT_\alpha(\FF,\VV)$ is closed under rewriting by $\RR$, for
each $\alpha \in S$.
\end{remark}

The following theorem states that confluence is a persistent property of TRSs.
\begin{theorem}
\label{thm-persistence}
Let a TRS $\RR$ be compatible with a sort attachment $\SS$. 
If $\RR$ is confluent on $\TT_\SS(\FF,\VV)$ then $\RR$ is confluent.
\end{theorem}

\begin{proof}
Assume that $\RR$ is confluent on $\TT_\SS(\FF,\VV)$.
We let $\bL$ be the smallest set such that
$\TT_\SS(\FF,\VV) \subseteq \bL$ and $\bL$ is closed under
replacing variables by holes and vice versa (cf.~\cvar).
It is easy to see that 
$\RR$ is layered according to $\bL$.
\crewrite and \cconsistent follow from the compatibility
assumption and Remark~\ref{rem-sort-closed}.
Also \cstepfusion is confirmed easily.
We show that $\RR$ is confluent on terms of rank one. To this end,
consider a term $s \in \bLT$. 
The confluence assumption on $\TT_\SS(\FF,\VV)$ does not immediately
apply to~$s$ since the variables need not match the type of their context.
If $s$ is a variable then $s$ is confluent. 
Otherwise, there is a term $s'$ in $\TT_\SS(\FF,\VV)$ that has $s$ as an
instance. Because
subterms of sort $\alpha$ are interchangeable in many-sorted terms,
we may choose $s'$ in such a way that $s'|_p = s'|_q$ if
$s'|_p, s'|_q \in \VV_\alpha$ for some $\alpha$ and
$s|_p = s|_q$. Note that for each $p$ the sort of $s'|_p$
is uniquely determined by $s$. Because the sets $\TT_\alpha(\FF,\VV)$
are pairwise disjoint, any rewrite sequence on $s \in \bLT$
is mirrored by a rewrite sequence from
$s' \in \TT_\SS(\FF,\VV)$. By assumption,
$s'$ is confluent and hence $s$ is confluent as well. We
conclude that $\RR$ is confluent on terms of rank one and hence
confluent by Theorem~\ref{thm-main-con}.
\end{proof}

\section{Order-sorted Persistence}
\label{app-order-sorted2}

In this section we establish order-sorted persistence. 
Section~\ref{sec-order-cr} introduces order-sorted rewriting, states
the main result, and explains how to exploit it for establishing
confluence.
In Section~\ref{app-order-sorted1} we prove the result for left-linear 
systems before Section~\ref{sec-restricted-ls} shows that layer systems
cannot immediately cover arbitrary TRSs.
We refine them such that they become suitable and give an alternative
proof for many-sorted persistence (Section~\ref{sec-persistence-r})
before we finally prove order-sorted persistence in 
Section~\ref{sec-order-proof}.
We compare our result with the earlier result from~\cite{AT96} in 
Section~\ref{rel-order-sorted}.

\subsection{Confluence via Order-sorted Persistence}
\label{sec-order-cr}

To obtain order-sorted terms, we equip a set of sorts
$S$ with a precedence $\succ$ and
modify Definition~\ref{def-typed-terms} as follows.

\begin{definition}
\label{def-order-sorted}
Let $\SS$ be a sort attachment.
We define terms of sort $\alpha$ inductively by
$\TT_\alpha(\FF,\VV) = \VV_\alpha \cup \{
f(\seq{t}) \mid
f : \alpha_1 \times \cdots \times \alpha_n \to \alpha$,
${t_i \in \TT_{\beta_i}(\FF,\VV)}$,
${\alpha_i \succeq \beta_i}$, and
${1 \leq i \leq n} \}$.
The set of order-sorted terms is
$\TT_\SS(\FF,\VV) = \bigcup_{\alpha \in S} \TT_\alpha(\FF,\VV)$.
A term $t$ is \emph{strictly order-sorted} if 
$\rt(t|_p) : \alpha_1 \times \dots \times \alpha_n \to \alpha$
and $t|_{pi} \in \VV_{\beta}$ imply $\alpha_i = \beta$,
for all $p \in \Pos_\FF(t)$.
\end{definition}

Note that we obtain  many-sorted terms by letting ${\succ} = \varnothing$.
Next we define when a TRS is \emph{compatible} with a
sort attachment $\SS$ in the order-sorted setting.

\begin{definition}
\label{def-os-compat}
A TRS~$\RR$ is \emph{compatible} with a sort attachment $\SS$ if 
each rule $\ell \to r \in \RR$ satisfies condition~\eqref{os:1},
and \emph{strongly compatible} with $\SS$ if condition~\eqref{os:2}
is satisfied as well.
\begin{enumerate}
\item
\label{os:1}
If $\ell \in \TT_{\alpha}(\FF,\VV)$ and $r \in \TT_{\beta}(\FF,\VV)$ then
$\alpha \succeq \beta$ and $\ell$ is strictly order-sorted.
\item
\label{os:2}
If $r \in \VV_{\beta}$ then $\beta$ is maximal in~$S$.
If $r \notin \VV$ then $r$ is strictly order-sorted.
\end{enumerate}
\end{definition}

Note that condition \eqref{os:1} ensures that well-typed terms
are closed under rewriting.
The main result on order-sorted persistence is stated below.

\begin{theorem}
\label{thm-order}
Let $\RR$ be compatible with a sort attachment~$\SS$.
Furthermore assume that $\RR$ is left-linear, bounded duplicating,
or strongly compatible with~$\SS$.
If $\RR$ is confluent on $\TT_\SS(\FF,\VV)$ then it is confluent.
\end{theorem}

Theorem~\ref{thm-order} gives rise to a decomposition result
(presented in~\cite{AT96,AT97}) based on 
order-sorted persistence.
The decomposition is based on the observation that the
sort of a term restricts the rules that can be applied when rewriting
it; therefore we can decompose a TRS $\RR$ that is compatible with
a sort attachment $\SS$ into several TRSs $\RR_\alpha$ ($\alpha \in S$)
each containing the rules applicable to terms of sort $\alpha$ or less.
Formally, we define $\trianglerighteq$ on sorts as the
smallest transitive relation such that
${\succ} \subseteq {\trianglerighteq}$ and
$\alpha \trianglerighteq \alpha_i$ whenever
$f : \alpha_1 \times \cdots \times \alpha_n \to \alpha$,
and then define $\RR_\alpha = \{ \ell \to r \mid
\text{$\ell \to r \in \RR$, $\ell \in \TT_{\beta}(\FF,\VV)$, and
$\alpha \trianglerighteq \beta$} \}$.

The next example shows that order-sorted persistence is more
powerful than many-sorted persistence for decomposing TRSs.

\begin{example}[(adapted from \cite{AT96})]
Consider the TRS $\RR$ consisting of the rewrite rules
\begin{xalignat*}{4}
(1)\,\,\m{f}(x,\m{a}) &\to \m{g}(x) &
(2)\,\,\m{f}(x,\m{f}(x,\m{b})) &\to \m{b} &
(3)\,\,\m{g}(\m{c}) &\to \m{c} &
(4)\,\,\m{h}(x) &\to \m{h}(\m{g}(x)) 
\end{xalignat*}
and the set of sorts $S = \{ 0, 1, 2 \}$ with
$1 \succeq 0$. Let the sort attachment be given by
$\m{a}, \m{b} : 1$, $\m{c} : 0$,
$\m{f} : 0 \times 1 \to 1$,
$\m{g} : 0 \to 0$,
$\m{h} : 0 \to 2$, and $x : 0$.
It is straightforward to check that $\RR$ is consistent with $\SS$.
In the order-sorted TRS, only rules (1), (2), and (3)
can be applied to terms of sort $1$ and their reducts, rules
(3) and (4) can be applied to terms
of sort $2$, and only rule (3) can be applied to terms of sort $0$.
Hence, since
$\RR_1 = \{ (1), (2), (3) \}$
(which is terminating and has no critical pairs),
$\RR_2 = \{ (3), (4) \}$
(which is orthogonal),
and $\RR_0 = \{ (3) \}$ 
(orthogonal) are confluent, $\RR$ is confluent.
No such decomposition can be obtained with many-sorted
persistence. Consider a \emph{most general} sort attachment
making all rules many-sorted:
$\m{a}, \m{b}, \m{c}, x : 0$,
$\m{f} : 0 \times 0 \to 0$,
$\m{g} : 0 \to 0$, and
$\m{h} : 0 \to 2$.
Since terms of sort $2$ can have subterms of sort $0$, no
decomposition is possible.
\end{example}

The weaker conditions in Definition~\ref{def-os-compat}
for left-linear TRSs are beneficial.

\begin{example}
\label{ex-mot-order}
Consider the TRS $\RR$ consisting of the rewrite rules
\begin{xalignat*}{3}
\m{f}(\m{a}) &\to \m{f}(\m{f}(\m{h}(\m{c}))) &
\m{g}(\m{b}) &\to \m{g}(\m{g}(\m{h}(\m{c}))) &
\m{h}(x) &\to x
\end{xalignat*}
and the set of sorts $S = \{ 0, 1, 2 \}$ with $1, 2 \succeq 0$.
Let the sort attachment be given by
$\m{a} : 1$, $\m{b} : 2$, $\m{c}, x : 0$,
$\m{f} : 1 \to 1$, $\m{g} : 2 \to 2$, and
$\m{h} : 0 \to 0$,
Note that $\RR$ is compatible with $\SS$.
We can decompose $\RR$ into the component induced by sort $1$:
$\RR_1 = \{
\m{f}(\m{a}) \to \m{f}(\m{f}(\m{h}(\m{c}))),
\m{h}(x) \to x
\}$,
sort $2$:
$\RR_2 = \{
\m{g}(\m{b}) \to \m{g}(\m{g}(\m{h}(\m{c}))),
\m{h}(x) \to x
\}$,
and sort $0$:
$\RR_0 = \{ \m{h}(x) \to x \}$.
If we add the restrictions for non-left-linear systems,
the collapsing rule
$\m{h}(x) \to x$ enforces $\m{h} : \alpha \to \alpha$ for 
a maximal sort $\alpha$. Hence also the arguments of $\m{f}$ and $\m{g}$
have sort $\alpha$, and $\alpha$ is greater than or equal to the sort of
$\m{a}, \m{b}, \m{c}, \m{f}(x), \m{g}(x)$.
So the component induced by $\alpha$ contains all rules.
\end{example}

\subsection{Order-sorted Persistence for Left-linear Systems}
\label{app-order-sorted1}

In this section we show that layer systems can establish order-sorted
persistence for left-linear TRSs.

\begin{theorem}
\label{thm-order-ll}
Let $\RR$ be compatible with a sort attachment $\SS$.
If $\RR$ is left-linear and confluent on $\TT_\SS(\FF,\VV)$ then it
is confluent.
\end{theorem}
\begin{proof}
Let $\bL$ be the smallest set such that $\TT_\SS(\FF,\VV) \subseteq \bL$
and $\bL$ is closed under~\cvar. First we show that 
$\RR$ is weakly layered according to~$\bL$.
In the sequel we call contexts \emph{weakly order-sorted} if they
are order-sorted except that arbitrary variables may occur at any
position. (These are exactly the elements of $\bL$ and weakly
order-sorted terms are those in $\bLT$.)

Condition \ctop holds trivially and condition \cvar holds by assumption.
For \cpartial we assume that $L|_p \merge N = N'$ with $p \in \Pos_\FF(L)$
is defined. 
Since $L,N \in \bL$ obviously $N'$ is weakly order-sorted and so is
$L[N']_p$ since $\rt(L|_p) = \rt(N')$ and hence $L[N']_p \in \bL$.
The final condition is \crewrite. So let $s \to_{p,\ell\to r} t$
with $p \in \Pos_\FF(M)$ for the max-top $M$ of $s$.
We have $\rt(M|_p) = \rt(\ell)$ and hence
$M[\ell]_p$ is a layer. Since $M$ is the max-top of $s$ and $\ell$ is
left-linear there is a substitution $\sigma$ such that
${M[\ell\sigma]_p = M}$.
Hence ${M \to_{p,\ell\to r} M[r\sigma]_p}$.
By compatibility with the sort attachment~$\SS$ we have $r\sigma \in \bL$.
Furthermore if $\alpha$ and $\beta$ are the sorts of $\ell$ and $r$
then $\alpha \succeq \beta$ ensures that $M[r\sigma]_p$ is 
weakly order-sorted and hence a member of $\bL$.

Next we show confluence of terms of rank one. To this end let 
$s \in \bL \cap \TT(\FF,\VV)$. Then there are a term
$s' \in \TT_\SS(\FF,\VV)$
and a variable substitution~$\chi$ such that $s = s'\chi$.
Let $t \lud{\RR}{*}{\from} s \to^*_\RR u$. By left-linearity of $\RR$
there are terms $t'$ and $u'$ with $t = t'\chi$ and $u = u'\chi$ such that
$t' \lud{\RR}{*}{\from} s' \to^*_\RR u'$. The confluence assumption on
$\TT_\SS(\FF,\VV)$ yields
$t' \to^*_\RR v' \lud{\RR}{*}{\from} u'$.
Hence $t = t'\chi \to^*_\RR v'\chi \lud{\RR}{*}\from u'\chi = t$.
We conclude by Theorem~\ref{thm-main-ll}.
\end{proof}

\subsection{Variable-restricted Layer Systems}
\label{sec-restricted-ls}

The following example shows that Theorem~\ref{thm-main-con} alone cannot
establish Theorem~\ref{thm-order} for TRSs which are neither 
left-linear nor bounded duplicating.

\begin{example}
\label{ex-restrict-mot}
Consider the set of sorts
$S = \{ 0, 1, 2, 3, 4 \}$, where $2 \succeq 0$ and $2 \succeq 1$.
The sort attachment $\SS$ is given by
\begin{xalignat*}{4}
u &: 0 &
v &: 1 &
\m f &: 3 \times 3 \to 4 &
\m h &: 2 \times 2 \times 0 \times 1 \to 3
\\
x &: 2 &
y &: 3 &
\m g &: 3 \to 3 &
\m a,\m b &: 4
\end{xalignat*}
and the TRS $\RR$ consists of the rules
\begin{xalignat*}{3}
\m f(y,y) &\to \m a &
\m f(y,\m g(y)) &\to \m b &
\m h(x,x,u,v) &\to \m g(\m h(u,v,u,v))
\end{xalignat*}
Then $\RR$ is confluent on $\TT_\SS(\FF,\VV)$ because it is locally
confluent and terminating on order-sorted terms, noting that $u$ and $v$
never represent equal terms due to sort constraints. However, if we
take $\bL$ to be the closure of $\TT_\SS(\FF,\VV)$ under \cvar then
the term $\m f(t,t)$ with $t = \m h(z,z,z,z)$ is not confluent because
${\m a \from \m f(t,t) \to \m f(t,\m g(t)) \to \m b}$.
Note that $\m f(t,t)$ is not order-sorted but contained in $\bL$.
Furthermore, observe that $\RR$ is layered according to~$\bL$.
Finally note that $\bL$ is the smallest
layer system with this property that contains $\TT_\SS(\FF,\VV)$.

\end{example}

The above example does not contradict Theorem~\ref{thm-order}
since $\RR$ is not strongly
compatible with~$\SS$;
the right-hand sides of $\RR$ are not strictly order-sorted although
$\RR$ is neither left-linear nor bounded duplicating. In particular
we have an infinite reduction
$\m h(z,z,\Diamond(z),\Diamond(z))
\to_\RR \m g(\m h(\Diamond(z),\Diamond(z),\Diamond(z),\Diamond(z)))
\to^+_{\Diamond(x) \to x} \m g(\m h(z,z,\Diamond(z),\Diamond(z))) 
\to_\RR \cdots$\,.

The problem is that layer systems allow to replace variables by variables
of a different sort and hence contain terms which are not order-sorted,
enabling new rewrite steps (which does never happen in the many-sorted
case nor for left-linear systems in the order-sorted setting). 
Since $\TT_\SS(\FF,\VV) \subsetneq \bLT$, we have to study when confluence
on $\TT_\SS(\FF,\VV)$ implies confluence on $\bLT$ in order to apply 
Theorem~\ref{thm-main-con}.
Instead of proving the missing implication directly, we again pursue
a general approach. To this end we relax condition~\cvar such that 
variables need not be replaced by variables of different sort,
to enable the representation of $\TT_\SS(\FF,\VV)$ as $\bLT$,
where $\bL$ satisfies the following refined notion of layer systems.

\begin{definition}
\label{def-restrict}
Recall the conditions from Definition~\ref{def-laysys}.
We introduce the following condition:
\begin{enumerate}
\item[\cvarp]
  If $C[x]_p \in \bL$ then $C[\square]_p \in \bL$.
  If $C[\square]_p \in \bL$ then 
  $\{ x \in \VV \mid C[x]_p \in \bL \}$ is an infinite set.
\end{enumerate}
We call $\bL \subseteq \CC(\FF,\VV)$ a
\emph{variable-restricted layer system} if it satisfies the 
conditions \ctop, \cvarp, and \cpartial.
Analogously, a variable-restricted layer system
\emph{weakly layers} $\RR$ if \crewrite is satisfied and
\emph{layers} $\RR$ if \crewrite, \cconsistent, and \cstepfusion
are satisfied.
\end{definition}

To distinguish between variable-restricted and (unrestricted) layer
systems we denote the former by $\vL$ in the future.
Note that \cvar implies \cvarp, hence any layer system is
also a variable-restricted layer system.
Furthermore, for each variable-restricted layer system $\vL$ there
is a corresponding (unrestricted) layer system
$\bLv = \vL \cup \{ C[x]_p \mid \text{$C[\square]_p \in \vL$ and
$x \in \VV$} \}$.
Obviously $\vL \subseteq \bLv$.

With the new condition \cvarp it is now possible to adequately represent 
$\TT_\SS(\FF,\VV)$ by a variable-restricted layer system.

\begin{example}[(Example~\ref{ex-restrict-mot} revisited)]
\label{ex-restrict-mot2}
To obtain a variable-restricted layer system,
let $\vL$ be the smallest set such that
$\TT_\SS(\FF,\VV) \subseteq \vL$ and $\vL$ is closed under
replacing variables by holes. Then it satisfies \cvarp.
Note that $\vLT = \TT_\SS(\FF,\VV)$
and hence $t = \m{h}(z,z,z,z) \notin \vL$ and thus $\m{f}(t,t) \notin \vL$.
\end{example}

For a weakly layered TRS the reduct of a rank one term
again is a rank one term.

\begin{lemma}
\label{lem-bl-r-closed}
Let\/ $\vL$ be a variable-restricted layer system that weakly layers
a TRS $\RR$. Then $\vLT$ is closed under rewriting by $\RR$.
\end{lemma}

\begin{proof}
Let $t \in \vLT$ and $t \to_\RR u$. Note that $t$ is its
own max-top. By~\crewrite, its reduct $u$ is a layer and hence
$u \in \vLT$.
\end{proof}

In the remainder of this section we show the analogues of
Theorems~\ref{thm-main-ll}, \ref{thm-main-bd}, and \ref{thm-main-con} for
variable-restricted layer systems (cf.\ Corollary~\ref{cor-rmain}).

The case of left-linear systems is straightforward.

\begin{lemma}
\label{lem-restrict-ll}
Let $\vL$ be a variable-restricted layer system that
weakly layers a left-linear TRS $\RR$.
If $\RR$ is confluent on $\vLT$ then
$\RR$ is confluent on $\bLvT$.
\end{lemma}

\begin{proof}
Let $s \in \bLvT$. By \cvar and \cvarp
a term $s' \in \vLT$ and a variable substitution $\chi$
exist such that $s'\chi = s$.
Now consider rewrite sequences
$t \lud{\RR}{*}\from s \to_\RR^* u$. Thanks to left-linearity,
there are terms $t'$ and $u'$ with ${t'\chi = t}$, ${u'\chi = u}$, and
$t' \lud{\RR}{*}\from s' \to_\RR^* u'$. 
By repeated application of
Lemma~\ref{lem-bl-r-closed}, $t'$, $u'$ as well as all intermediate
terms are elements of $\vLT$. From the assumption we obtain a valley
${t' \to_\RR^* v' \lud{\RR}{*}\from u'}$, inducing a valley 
${t = t'\chi  \to_\RR^* v'\chi \lud{\RR}{*}\from u'\chi = u}$.
Note that $v'\chi \in \bLv$ (by Lemma~\ref{lem-bl-closed}) and obviously
$v'\chi \in \TT(\FF,\VV)$.
\end{proof}

To prepare for a result concerning bounded duplicating TRSs, we
generalize bounded duplication to weakly bounded duplication,
which turns out to be more suitable for the proof of
Lemma~\ref{lem-restrict-wbd} below.

\begin{definition}
\label{def-wbd}
We call $\RR$ \emph{weakly bounded duplicating} if
$\{ \top \to \bot \} / \RR$ is terminating for fresh constants $\top$
and $\bot$.
\end{definition}

\begin{lemma}
\label{lem-bd-wbd}
Any bounded duplicating TRS is weakly bounded duplicating.
\end{lemma}

\begin{proof}
Assume that $\RR$ is not weakly bounded duplicating.
So there exists an infinite rewrite sequence
$t_0 \to t_1 \to \cdots$ in 
$\RR \cup \{ \top \to \bot \}$ that contains infinitely many
applications of the rule $\top \to \bot$.
Let $t_i'$ be obtained from $t_i$ by replacing all occurrences of
$\top$ by $\Diamond(\bot)$. Since $\top$ does not appear in the rules
of $\RR$, we obtain an infinite rewrite sequence
$t_0' \to t_1' \to \cdots$ in 
$\RR \cup \{ \Diamond(x) \to x \}$ with infinitely many
applications of the instance $\Diamond(\bot) \to \bot$ of
$\Diamond(x) \to x$. Hence $\RR$ is not bounded duplicating.
\end{proof}

To see that weakly bounded duplication generalizes bounded duplication,
consider the TRS $\RR$ consisting of the single rule
$\m{f}(\m{a},x) \to \m{f}(x,x)$, which is
not bounded duplicating since
$\m{f}(\m{a},\Diamond(\m{a})) \to_\RR
\m{f}(\Diamond(\m{a}),\Diamond(\m{a})) \to_{\Diamond(x) \to x}
\m{f}(\m{a},\Diamond(\m{a})) \to_\RR
\cdots$, but weakly bounded duplicating.

Below we will establish the following two lemmata.

\begin{lemma}
\label{lem-restrict-wbd}
Let $\vL$ be a variable-restricted layer system that 
weakly layers a weakly bounded duplicating TRS~$\RR$.
If $\RR$ is confluent on $\vLT$
then $\RR$ is confluent on $\bLvT$.
\end{lemma}

\begin{lemma}
\label{lem-restrict-con}
Let $\vL$ be a variable-restricted layer system that 
layers a TRS $\RR$.
If $\RR$ is confluent on $\vLT$ then $\RR$ is confluent on $\bLvT$.
\end{lemma}

For both proofs we are given a variable-restricted layer system
$\vL$ that weakly layers a TRS $\RR$.
We fix an initial term $s \in \bLvT$ and show that it is confluent.
Since $\vL \subseteq \bLv$ the confluence assumption on $\vLT$ may
not apply to $s$. To overcome this problem we use \cvar
and \cvarp to construct a term $s' \in \vLT$ and a variable substitution
$\chi$ such that $s = s'\chi$ and fix a well-order
$\gg$ on $\Var(s')$. We extend $\gg$ to terms by closing
it under contexts and transitivity.

\begin{center}
\fbox{\emph{Let $s \in \bLvT$, $s' \in \vLT$, and $\chi$ 
with $s = s'\chi$ be fixed.}}
\end{center}

\begin{definition}
A term $t' \in \vLT$ is a \emph{representative} of $t \in \bLvT$ if
$t = t'\chi$ and $\Var(t') \subseteq \Var(s')$. A representative
$t'$ of $t$ is called \emph{minimal} if it is minimal with respect to
$\gg$.
\end{definition}

Note that $s'$ is a representative of $s$.
Before proving key properties for representatives we show how they
help to avoid the situation of Example~\ref{ex-restrict-mot}.

\begin{example}[(Example~\ref{ex-restrict-mot} revisited)]
\label{ex-restrict-mot3}
Consider the variables with sorts
\begin{align*}
x_1, x_2, x_5, x_6: 2  &&
x_3, x_7: 0 &&
x_4, x_8: 1
\end{align*}
and order $x_8 \gg x_7 \gg \cdots \gg x_1$. The term
$s = \m{f}(t,t) \in \bLvT$ has the representative
$s' = \m{f}(\m{h}(x_1,x_2,x_3,x_4),\m{h}(x_5,x_6,x_7,x_8)) \in \vLT$
and the (unique) minimal representative
$\lst{s} = \m{f}(\lst{t},\lst{t}) \in \vLT$
where $\lst{t} = \m{h}(x_1,x_1,x_3,x_4)$.
The peak $\m{a} \from \m{f}(t,t) \to \m{f}(t,\m{g}(t)) \to \m{b}$ in
$\bLvT$ 
is simulated by
\[
\m{a} \from \m{f}(\lst{t},\lst{t}) \to
\m{f}(\lst{t},\m{g}(\m{h}(x_3,x_4,x_3,x_4))) \gg
\m{f}(\lst{t},\m{g}(\lst{t})) \to \m{b}
\]
in $\vLT$. Note that the $\gg$-step replaces
$\m{f}(\lst{t},\m{g}(\m{h}(x_3,x_4,x_3,x_4)))$
by the least representative $\m{f}(\lst{t},\m{g}(\lst{t}))$ of
$\m{f}(t,\m{g}(t))$.
\end{example}

The key operation on representatives and related terms is
copying variables between them, as justified by the following
lemma.

\begin{lemma}
\label{lem-repr-copy}
Let $L, N \in \vL$ be layers with $L_\square = N_\square$.
If $p \in \Pos_{\VC}(L)$ then $L[N|_p]_p \in \vL$.
\end{lemma}

\begin{proof}
If $p = \epsilon$ then the claim is trivial. Otherwise, let
$L' = L[\square]_p$ and
$N' = N[\square]_{q \in \Pos_{\VC}(L) \setminus \{ p \}}$.
We have $L', N' \in \vL$ by applications of property \cvarp and 
$L[N|_p]_p = L' \merge N'$ by assumption. Property~\cpartial
yields the desired $L[N|_p]_p \in \vL$.
\end{proof}

The next lemma establishes that the minimal representative (if it exists)
is unique, justifying the name \emph{least} representative.
The proof makes the construction in Example~\ref{ex-restrict-mot3}
explicit and is illustrated by Example~\ref{ex-repr}.

\begin{lemma}
\label{lem-repr}
If $t \in \bLvT$ has a representative then it has a least
representative.
\end{lemma}
\begin{proof}
We have to show the existence and uniqueness of a minimal representative
of $t$. From a representative $t'$ we obtain $t'_\square \in \vL$ using
\cvarp repeatedly. 
Consider $V_p = \{ x \in \Var(s') \mid \text{$\chi(x) = t|_p$ and 
$t'_\square[x]_p \in \vL$} \}$ for each $p \in \Pos_\VV(t')$.
Note that $t'|_p \in V_p$ because we can insert the variable $t'|_p$ into
$t'_\square$ at position $p$ by Lemma~\ref{lem-repr-copy} to obtain a
layer in $\vL$. Hence $V_p$ is non-empty. Since it is also finite it has
a minimum element $\min(V_p)$ with respect to $\gg$. Let
$\dot t = t'_\square[\min(V_p)]_{p \in \Pos_\VV(t')}$.
We have $\dot t \in \vL$ by \cvarp and the definition of $V_p$.
Clearly $\dot t \in \TT(\FF,\VV)$
and $\Var(\dot t) \subseteq \Var(s')$ because all holes are replaced
by some variable from $\Var(s')$. Moreover, $\dot t \chi = t$ by
construction, in particular the definition of $V_p$. It follows that
$\dot t$ is a representative of $t$.
Note that $\dot t$ does not depend on the choice of $t'$
because  $t'_\square = t^{}_\square$. Therefore,
$t' \gg^= \dot t$ for any representative $t'$ of $t$,
which makes $\dot t$ the least representative of $t$. 
\end{proof}

\begin{example}[(Example~\ref{ex-restrict-mot3} revisited)]
\label{ex-repr}
Consider 
$s  = \m f(\m h(z,z,z,z),\m h(z,z,z,z))$ and 
$s' = \m f(\m h(x_1,x_2,x_3,x_4),\m h(x_5,x_6,x_7,x_8))$ with
$\chi(x_i) = z$ for all $1\leq i \leq 8$.
Then 
$s'_\square = \m f(\m h(\square,\square,\square,\square),
                   \m h(\square,\square,\square,\square))$.
Since $V_{11} = V_{12} = V_{21} = V_{22} = \{x_1,\ldots,x_8\}$,
$V_{13} = V_{23} = \{x_3,x_7\}$, and 
$V_{14} = V_{24} = \{x_4,x_8\}$ we obtain
$\dot{s} = \m f(\m h(x_1,x_1,x_3,x_4),\m h(x_1,x_1,x_3,x_4))$. 
\end{example}

We denote the least representative term of a representable term
$t \in \bLvT$ by $\lst{t} \in \vLT$.
The following lemma states that a rewrite step performed on a term
in $\bLv$ can be mirrored on its least representative in $\vL$.
Recall that in Example~\ref{ex-restrict-mot3} the representative $s'$ is
a normal form but the step from $s$ can be mirrored on $\lst{s}$.

\begin{lemma}
\label{lem-r-repr}
Let $t, u \in \bLvT$ with $t \to_\RR u$  such that $\lst{t}$ exists.
\begin{enumerate}
\item
\label{lem-r-repr-1}
If\/ $\vL$ weakly layers $\RR$
then $\lst{t} \to_\RR u'$ for some representative $u'$ of $u$.
\item
\label{lem-r-repr-2}
If\/ $\vL$ layers $\RR$ then $u' = \lst{u}$ or $u' \in \VV$
in \eqref{lem-r-repr-1}.
\end{enumerate}
\end{lemma}

\begin{proof}
\mbox{}
\begin{enumerate}
\item
\label{lem-r-repr-a}
Assume that $\lst{t}$ is the least representative of $t$ and
let $t \to_{p,\ell \to r} u$. We obtain a context
$C \in \vL$ by replacing all variables in $t$ by $\square$.
By Lemma~\ref{lem-context-instance},
there is a term $c$ with $C = c\sigma_\square$
and $\ell \matches c|_p$.
To ensure $c \matches \lst{t}$ we need to show
$\lst{t}|_q = \lst{t}|_r$ for all $x \in \Var(c)$ and
$q, r \in \Pos_x(c)$.
To that end, fix $x$ and let $P = \Pos_x(c)$.
For each $q \in P$, $\lst{t}|_q$ is a variable.
Let $y = \min\,\{ \lst{t}|_q \mid q \in P \}$. We
will show that
$\lst{t}|_q = y$ for all $q \in P$.
Consider the max-top $M \in \vL$ of $C[y,\dots,y]$.
Note that $c \matches C[y,\dots,y]$, so that
$\ell \matches C[y,\dots,y]|_p$.
From condition \crewrite we obtain $\ell \matches M|_p$
and thus $c \matches M$ by
Lemma~\ref{lem-context-instance}\eqref{lem-context-instance-D}
since $C \Cleq M$.
By construction $\lst{t}|_q = y$ for some $q \in P$.
Since $C[y]_q$ is a layer by Lemma~\ref{lem-repr-copy},
$M \merge C[y]_q$ is a layer according to \cpartial. Because
$M$ is the max-top of $C[y,\dots,y]$, $M \merge C[y]_q = M$ and thus
$M|_q = y$.
It follows that $M|_q = y$ for all $q \in P$, since otherwise $M$ would
fail to be an instance of $c$.
Repeated applications of Lemma~\ref{lem-repr-copy} yields
$t' = \lst{t}[y]_{q \in P} \in \vL$.
We have $t' = \lst{t}$ by the choice of $y$ and the minimality of 
$\lst{t}$.
We conclude that $c \matches \lst{t}$ and hence
$\ell \matches \lst{t}|_p$, which induces
a rewrite step $\lst{t} \to_{p,\ell \to r} u'$ as claimed.
The term $u'$ is a representative of $u$ because $u'\chi = u$,
$u' \in \vL$ by Lemma~\ref{lem-bl-r-closed}, and rewriting does not
introduce variables.
\item
\label{lem-r-repr-b}
Assume that $u'$ is not a least representative of $u$. 
We have $u' \gg \lst{u}$, so there is a position $q \in \Pos_\VV(u)$
with $z = u'|_q \gg \lst{u}|_q = y$. Let $C = c\sigma_\square$ as
in the proof of part (1).
There is a rewrite step $c \to_{p,\ell \to r} d$ for some term $d$
and $C \to_{p,\ell \to r} D = d \sigma_\square$.
Let $M \in \vL$ and $L \in \vL$ be the max-tops of $C_y = C[y,\dots,y]$
and $D_y = D[y,\dots,y]$. Note that
$C_y \to_{p,\ell \to r} D_y$,
which implies $M \to_{p,\ell \to r} L$ by
\cconsistent except when $M \to_{p,\ell \to r} \square$.
In the latter case $r$ and thus also $u'$ is a variable, and we are
done. So assume $M \to_{p,\ell \to r} L$.
Consider the variable $x = d|_q$. We must have $L|_q = y$
because otherwise we could copy $\lst{u}|_q = y$ to $L$ by
Lemma~\ref{lem-repr-copy}.
The term $\lst{t}$ and the context $M$ are instances 
of $c$ and so there are substitutions $\sigma_{\lst{t}}$ and $\sigma_M$
such that $c\sigma_{\lst{t}} = \lst{t}$ and
$c \sigma_M = M$. We have $\sigma_{\lst t}(x) = u'|_q = z$ and
$\sigma_M(x) = y$
because $d\sigma_M = L$. Since $x \in \Var(d)$ and
$c \to_\RR d$, the set $\Pos_x(c)$ is non-empty. Let $q' \in \Pos_x(c)$.
The layer $C[y]_{q'} \in \vL$ can
be obtained by copying $M|_{q'} = y$ to $C$ using
Lemma~\ref{lem-repr-copy}.
Since $\lst{t}|_{q'} = \sigma_{\hat t}(x) = z$, we obtain
$\lst{t} \gg \lst{t}[y]_{q'} \in \vL$. The term
$\lst{t}[y]_{q'}$ is a representative of $t$ because $\chi(y) = \chi(z)$
(recall that $u = u'\chi = \lst{u}\chi$). Hence we obtained a
contradiction with the minimality of $\lst{t}$.
\qed
\end{enumerate}
\end{proof}

The following lemma shows that instead of adding a single
rule $\top \to \bot$, we can extend a weakly bounded duplicating TRS
with any terminating ARS, where the objects are regarded as fresh
constants, and still obtain relative termination. The
induced well-founded order will be used in the proof of
Lemma~\ref{lem-restrict-wbd}.

\begin{lemma}
\label{lem-wbd}
Let $\RR$ be a weakly bounded duplicating TRS and $\AA$ a terminating
ARS. If $\RR$ and $\AA$ share no constants then $\AA$ is terminating
relative to $\RR$.
\end{lemma}

\begin{proof}
We use reduction pairs for this proof, which are pairs consisting
of a quasi-order $\succeq$ and a well-founded strict order $\succ$ that
are compatible:
${\succeq} \cdot {\succ} \cdot {\succeq} \subseteq {\succ}$.
Reduction pairs give rise to a multiset extension
in a straightforward way (e.g.,
the definitions of $\succ_\text{gms}$ and $\succeq_\text{gms}$ in
\cite{TAN12}).
We denote the objects in $\AA$ by $\OO$.
Let $\FF$ be the signature of $\RR$.
From the termination of $\AA$ we obtain a well-founded order $\succ$ on
$\OO$ such that $\AA \subseteq {\succ}$.
For each $\alpha \in \OO$ define a map $\pi_\alpha$ from
$\TT(\FF \cup \OO,\VV)$ to
$\TT(\FF \cup \{ \top, \bot \},\VV)$ as follows:
\[
\pi_\alpha(t) = \begin{cases}
\top & \text{if $t = \alpha$} \\
\bot & \text{if $t \in \OO \setminus \{ \alpha \}$} \\
f(\pi_\alpha(t_1),\dots,\pi_\alpha(t_n)) & \text{if $t = f(\seq{t})$ with
$f \in \FF$} \\
t & \text{if $t \in \VV$}
\end{cases}
\]
We measure terms by the set
$\# t = \{ (\alpha, \pi_\alpha(t)) \mid \alpha \in \Fun(t) \cap \OO \}$.
The measures of two terms are compared by the multiset extension of
the lexicographic product of the precedence $\succ$
on $\OO$ and the reduction pair consisting
of the well-founded (by the weakly bounded termination assumption)
order $\to_{\{ \top \to \bot \} /\RR}^+$ and
the compatible quasi-order $\to_\RR^*$.
Each application of a rule $\alpha \to \beta$ from $\AA$ decreases
the component associated with $\alpha$ in $\# t$ and introduces
or modifies a component associated with $\beta$ in $\# t$,
giving rise to a decrease in the strict part of the multiset
extension. Moreover, if $t \to_\RR u$ then
$\pi_\alpha(t) \to_\RR \pi_\alpha(u)$, for all $\alpha \in \OO$.
Hence the terms are related by the non-strict part of the multiset
extension.
It follows that $\AA$ is terminating relative to $\RR$.
\end{proof}

\begin{proof}[of Lemma~\ref{lem-restrict-wbd}]
To show confluence of $s$ we introduce a relation $\RRV{}$ that 
allows to map an $\RR$-peak from $s$ to a $\RRV{}$-peak.
Afterwards we show confluence of $\RRV{}$ and
conclude by ${\RRV{}} \subseteq {\to^*_\RR}$.

We write $t \RRV[t_0']{} u$ if $t_0' \to_\RR^* \lst{t}$ and
$s \to_\RR^* t \to_\RR^* u$
such that $t \to_\RR^* u$ is mirrored by $\lst{t} \to_\RR^* u'$
with $u = u'\chi$.
Labels are compared using the order
${\succ} := {\to_{{\gg} / \RR}^+}$, which
is well-founded according to Lemma~\ref{lem-wbd}
applied to the ARS $(\Var(s'),{\gg})$, where we regard 
the elements of $\Var(s')$ as constants for this purpose.

First we show that a peak consisting of $\RR$-steps can be represented as
a peak of $\RRV{}$-steps. To this end we claim that
$t \RRV[\lst{t}]{} u$ whenever $s \to^*_\RR t \to_\RR u$.
To show the claim, note that $s$ has a least representative
by Lemma~\ref{lem-repr}, and that by
Lemmata~\ref{lem-r-repr}\eqref{lem-r-repr-a} and Lemma~\ref{lem-repr}
each
immediate successor
of a term with a least representative also has
a least representative. Therefore, $t$ has a least representative,
and we conclude by another application of
Lemma~\ref{lem-r-repr}\eqref{lem-r-repr-a}.
Next we establish that $\RRV{}$ is locally decreasing and hence confluent
by Theorem~\ref{thm-dd}. Consider a local peak
$u \LLV[t_0']{} t \RRV[t_1']{} v$.
By definition of $\RRV{}$ there are representatives $u'$ and $v'$ of $u$
and $v$ such that $u' \lud{\RR}{*}\from \lst{t} \to_\RR^* v'$.
We obtain $u' \to_\RR^* w' \lud{\RR}{*}\from v'$ from the confluence
assumption on $\vLT$.
Consider the sequence $u' \to_\RR^* w'$. If $u' = \lst{u}$
then $u \RRV[t_1']{} w'\chi$, noting that
$t_1' \to_\RR^* \lst{t} \to_\RR^* u'$.
Otherwise, there is a rewrite
sequence $u = u' \chi = u_1 \to_\RR \cdots \to_\RR u_n = w'\chi = w$,
such that $u' \gg \lst{u} \to_{{\gg}\cup\RR}^* \lst{u}_i$ and
thus $u' > \lst{u}_i$ for all $1 \leqslant i \leqslant n$.
Hence we obtain $u \RRV[\curlyvee t_1']{*} w$ by repeated use
of the above claim. Analogously, we obtain $v \RRV[t_0']{} w$ or
$v \RRV[\curlyvee t_0']{*} w$. The proof is concluded by the obvious
observation that ${\RRV{}} \subseteq {\to^*_\RR}$.
\end{proof}

\begin{proof}[of Lemma~\ref{lem-restrict-con}]
Consider a peak $t \lud{\RR}{*}\from s \to_\RR^* u$. Obviously $s$ has
a representative and hence also a least representative $\lst{s}$ by 
Lemma~\ref{lem-repr}. Using Lemma~\ref{lem-r-repr}
repeatedly we obtain a
peak $t' \lud{\RR}{*}\from \lst{s} \to_\RR^* u'$, noting
that all
reducts
of $\lst{s}$ are
least representatives of the corresponding reducts of $s$ or
variables, but since variables are normal forms
the latter can only happen in the last step.
From the confluence assumption on $\vLT$ we obtain
$t' \to_\RR^* v' \lud{\RR}{*}\from u'$.
Applying the variable substitution $\chi$ yields
$t = t'\chi \to_\RR^* v'\chi \lud{\RR}{*}\from u'\chi = u$ on $\bLvT$.
\end{proof}

\begin{lemma}
\label{lem-c-preserve}
If a TRS is (weakly) layered according to a variable-restricted layer
system then it is (weakly) layered according to the corresponding
(unrestricted) layer system.
\end{lemma}
\begin{proof}
The result for weakly layered TRSs is obvious.
The result for layered TRSs follows from Lemma~\ref{lem-r-repr}.
\qed
\end{proof}

\begin{corollary}
\label{cor-rmain}
The statements of Theorems~\ref{thm-main-ll},~\ref{thm-main-bd},
and~\ref{thm-main-con}
remain true when based on a variable-restricted layer system.
\end{corollary}
\begin{proof}
In case of 
left-linear TRSs we conclude by Theorem~\ref{thm-main-ll} and
Lemmata~\ref{lem-restrict-ll} and \ref{lem-c-preserve}.
For bounded duplicating TRSs we use Theorem~\ref{thm-main-bd} and
Lemmata~\ref{lem-bd-wbd}, \ref{lem-restrict-wbd}, and
\ref{lem-c-preserve}. For TRSs that are layered according to a 
variable-restricted layer system we use Theorem~\ref{thm-main-con} and 
Lemmata~\ref{lem-restrict-con} and~\ref{lem-c-preserve}.
\end{proof}

\subsection{Many-sorted Persistence by Variable-restricted Layer Systems}
\label{sec-persistence-r}

We demonstrate the usefulness of variable-restricted layer systems
by the following alternative proof of Theorem~\ref{thm-persistence},
which avoids the complication of establishing confluence on $\bLT$.

\begin{proof}[of Theorem \ref{thm-persistence}]
Assume that $\RR$ is confluent on $\TT_\SS(\FF,\VV)$.
We let~$\vL$ be the smallest set such that
$\TT_\SS(\FF,\VV) \subseteq \vL$ and $\vL$ 
is closed under replacing variables by holes. So $\vL$ trivially
satisfies \cvarp.
Hence $\vLT = \TT_\SS(\FF,\VV)$ and thus $\RR$ is confluent on $\vLT$
by the assumption.
It is easy to see that $\vL$ is a variable-restricted layer system
layering $\RR$; conditions \crewrite and \cconsistent follow from the
compatibility assumption. Therefore $\RR$ is confluent by 
Corollary~\ref{cor-rmain}.
\end{proof}

\subsection{Order-sorted Persistence by Variable-restricted Layer Systems}
\label{sec-order-proof}

In this section we prove the main result on order-sorted persistence.

\begin{proof}[of Theorem~\ref{thm-order}]
Assume that $\RR$ is compatible with $\SS$.
To define layers as order-sorted terms, we add a fresh, minimum sort
$\bot$ with $\square : \bot$ and require that no variable has sort $\bot$.
The set $\vL := \TT_{\SS \cup \{ \bot \}}(\FF \cup \{ \square \}, \VV)$
is a variable-restricted layer system that satisfies \cstepfusion.

We show that $\vL$ satisfies condition \crewrite.
So let $M$ be the max-top of $s$, $p \in \Pos_\FF(M)$, and
$s \to_{p,\ell \to r} t$.
Because $\ell$ is order-sorted,
$\Pos(\ell) \subseteq \Pos(M|_p)$.
We claim that $\ell \matches M|_p$. If
$\ell|_q = \ell|_{q'} \in \VV_\alpha$ then
$M|_{pq} = M|_{pq'} \in \TT_{\alpha'}(\FF \cup \{ \square \},\VV)$ for
some $\alpha'$ with $\alpha \succeq \alpha'$, due to the fact that
$\ell$ is strictly order-sorted.
Let $\sigma$ be a substitution such that $\ell\sigma = M|_p$.
Using the compatibility condition (of
Definition~\ref{def-order-sorted}), we readily obtain
$L = M[r\sigma]_p \in \vL$.

Next we show that if $\RR$ is strongly compatible with $\SS$, then
condition \cconsistent holds. So assume that
$\RR$ is neither left-linear nor bounded duplicating
and $L \neq \square$. We show that $L$ is the max-top of $t$.
Let $L'$ be the max-top of $t$. First of all,
if $r$ is not a variable and $\ell|_q = r|_{q'} \in \VV_\alpha$ then
$L'|_{pq'} = M|_{pq} = L|_{pq'}$
because $\ell$ and $r$ are strictly order-sorted.
This implies $L = L'$.
Next suppose that $r = x \in \VV_\beta$.
Let $p'$ be the position directly above $p$ and let
$\rt(L|_{p'}) : \beta_1 \times \dots \times \beta_n \to \beta'$.
We have $p = p'i$ for some $1 \leqslant i \leqslant n$. We claim
that $\beta_i = \beta$. Let $\alpha$ be the sort of $\ell$.
We have $\alpha \succeq \beta$ and $\beta_i \succeq \alpha$.
According to the second compatibility condition, $\beta$ is maximal in
$S$ and thus $\beta = \alpha = \beta_i$.
It follows that $L'|_p = M|_{pq} = L|_p$ for any $q \in \Pos_x(\ell)$.

Note that $\vLT = \TT_\SS(\FF,\VV) = \vL \cap \TT_\SS(\FF,\VV)$.
The proof is concluded with an appeal to Corollary~\ref{cor-rmain}.
\end{proof}

\section{Related Work}
\label{sec-related}

As we already mentioned in the introduction, modularity of term rewrite
systems has been reproved several times. A number of related results
have been proved by adapting the proof of \cite{KMTV94} and there have
been several previous attempts to make the result more reusable.
\cite{O94a} casts the modularity result in terms of a
collapsing reduction $\to_c$, and shows that for
composable TRSs, confluence is modular if $\to_c$ is normalizing.
Toyama's theorem arises as a special case.
\cite{K95} proposes an abstract framework, based on
so-called \emph{pre-confluences} and \emph{context selectors}
constructed from pre-confluences. The latter can be seen as
a precursor of layer systems. In particular, the selection of
max-tops gives rise to a (proper) context selector. However, the
notion of pre-confluences is geared towards the uncurrying
application, and too restrictive to encompass modularity of
confluence \cite{K11}.
A third approach to abstraction is taken in \cite{L96}.
In this work, modularity of confluence is proved using category
theory, exploiting the fact that terms can conveniently be modeled
by a monad. Unfortunately, the development is flawed, and only
applies to TRSs over unary function symbols and constants.%
\footnote{%
\def\X{\acute\ }
The paper claims that for any TRS $\Theta$, the monad
$T_\Theta$ is \emph{strongly finitary}, which implies
that it preserves coequalizers. This is not true in general. As
an example, let $\star$ be the trivial category and consider the
coequalizer $Q : \star + \star \to \star$
of the injections $\iota_1, \iota_2 : \star \to \star + \star$.
Furthermore let $\Theta = \{ \m{f}(x,x) \to x \}$.
Then $T_{\Theta}(Q)$ equates $\m{f}(\X\iota_1\star,\X\iota_1\star)$
and $\m{f}(\X\iota_1\star,\X\iota_2\star)$, but the coequalizer of
$T_\Theta(\iota_1)$ and $T_\Theta(\iota_2)$ does not,
because $\m{f}(\X\iota_1\star,\X\iota_1\star)$ is not in the image of
either of these functors.}

In the remainder of this section we discuss specific issues,
starting with a comparison of our result on order-sorted persistence 
to~\cite{AT96} in Section~\ref{rel-order-sorted}.
In Section~\ref{rel-modularity} we reflect on the differences
between \cite{KMTV94} and \cite{vO08b}, which correspond to
changes from the earlier conference paper \cite{FZM11} to the present
article.
In Section~\ref{rel-constructivity} we elaborate upon the constructivity
claim made in Section~\ref{sec-introduction}.

\subsection{Order-sorted Persistence}
\label{rel-order-sorted}

In this section we compare our result from
Section~\ref{app-order-sorted2} to the main result of~\cite{AT96},
which can be stated as follows.

\begin{definition}
\label{def-AT96}
A sort attachment $\SS$ is \emph{\compatat} with a TRS $\RR$ if
condition ($\star$) is satisfied for each rewrite rule
$\ell \to r \in \RR$:
\begin{enumerate}
\item[($\star$)]
If $\ell \in \TT_\alpha(\FF,\VV)$ and $r \in \TT_\beta(\FF,\VV)$ then
$\alpha \succeq \beta$ and $\ell$, $r$ are strictly order-sorted.
\end{enumerate}
\end{definition}

The main claim in~\cite{AT96} is that Theorem~\ref{thm-order} holds
for \compatat systems. We show that this is incorrect.
The counterexample 
presented here is simpler than our previous example in~\cite{FZM11}.

\begin{example}
\label{ex-counterexample}
We use $\{ 0,1,2,3 \}$ as sorts where $1 \succeq 0$ and sort attachment
$\SS$
\begin{xalignat*}{5}
x &: 0 &
\m{f} &: 0 \to 2 &
\m{h} &: 1 \times 0 \to 2 &
\m{e} &: 0 \to 1 &
\m{c} &: 1
\\
y &: 2 &
\m{g} &: 2 \to 2 &
\m{i} &: 2 \times 2 \to 3&
\m{a}, \m{b} &: 3
\end{xalignat*}
Consider the TRS $\RR$ consisting of the rules
\begin{xalignat*}{5}
\m{f}(x) &\to \m{h}(\m{e}(x), x) &
\m{h}(\m{c}, x) &\to \m{g}(\m{f}(x)) &
\m{e}(x) &\to x
&
\m{i}(y,y) &\to \m{a} &
\m{i}(y,\m{g}(y)) &\to \m{b}
\end{xalignat*}
This TRS is \compatat with $\SS$. 
On order-sorted terms
it is locally confluent and terminating and thus confluent
(note that $x$ may not be instantiated
by $\m{c}$ due to the sort constraints). It is not confluent
on arbitrary terms
because
\[
\m{a} \from
\m{i}(\m{f}(\m{c}), \m{f}(\m{c})) \to^*
\m{i}( \m{f}(\m{c}),\m{g}(\m{f}(\m{c}))) \to \m{b}
\]
\end{example}

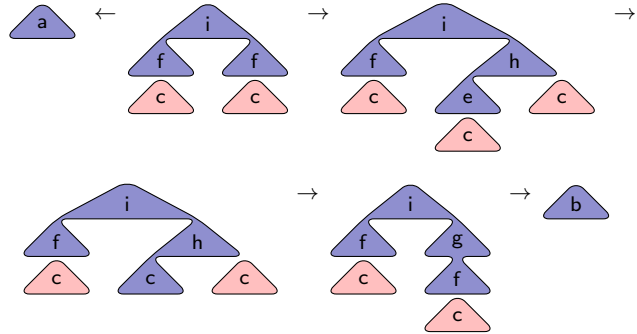
\begin{figure}[t]
\begin{gather*}
\Vtop{\input{samples/qa.tikz}} \from
\Vtop{\input{samples/qb.tikz}} \to
\Vtop{\input{samples/qc.tikz}} \to\\
\Vtop{\input{samples/qd.tikz}} \to
\Vtop{\input{samples/qe.tikz}} \to
\Vtop{\input{samples/qf.tikz}}
\end{gather*}
\caption{Non-confluence in Example~\ref{ex-counterexample}.}
\label{fig-counterexample}
\end{figure}

Note that any \compatat TRS is strongly compatible
(cf.\ Definition~\ref{def-os-compat}), unless it is neither left-linear
nor bounded duplicating, and contains a collapsing rule.
Indeed the TRS $\RR$ of Example~\ref{ex-counterexample}
has all these features.
Ultimately, the culprit is the collapsing rule $\m{e}(x) \to x$,
causing fusion from above (cf.\ Figure~\ref{fig-counterexample}).
This case is not considered in the proof of~\cite[Proposition~3.9]{AT96}.
Definition~\ref{def-os-compat} takes care of the problem with
collapsing rules in Definition~\ref{def-AT96}. Furthermore, it puts
fewer constraints on the right-hand sides in case of left-linear
or bounded duplicating systems, which is beneficial
(cf.\ Example~\ref{ex-mot-order}).

\subsection{Modularity}
\label{rel-modularity}

We compare the proof setups of \cite{KMTV94} and \cite{vO08b}.

The first difference concerns the decomposition of terms.
Whereas Klop et al.\ split a term into its max-top and aliens,
van Oostrom splits it into a base context and a sequence of
tall aliens. This is the key for making the proof constructive:
while fusion of an alien may cause many new aliens to appear,
none of them will be tall, so they do not have to be tracked
explicitly. In contrast, Klop et al.\ start by constructing
witnesses, and thus prevent aliens from fusing while establishing
confluence.

The other ingredients of the proofs are quite similar: The proof setup
is an induction on the rank of the starting term. One distinguishes
inner ($\to_i^*$, acting on aliens) and outer ($\to_o^*$,
acting on the max-top) steps (Klop et al.), or tall ($\rrr[\iota]{}$,
acting on the tall aliens) and short ($\RRR{}$, acting on the base
context)
steps (van Oostrom). One then argues as follows:
\begin{enumerate}
\item
Outer (short) steps are confluent because one can replace the
principal (tall) subterms by suitable variables in the top (base)
context, and then invoke the induction hypothesis.
\item
Inner (tall) steps are confluent because they only act on principal
subterms (tall aliens). In joining these subterms, one can ensure that
any equalities between them are preserved (we call such sequences of
inner steps \emph{balanced}). In van Oostrom's proof, the
resulting joining sequences may involve fusion and therefore short steps,
but by ranking short steps below tall steps,
a locally decreasing diagram is obtained.
\item
Balanced inner steps (tall steps) and outer steps (short steps)
commute (can be joined decreasingly). The idea is to replace the
principal subterms (tall aliens) of source and target of the inner steps
by the same variables, so that the outer steps can be simulated on the
result. In van Oostrom's proof, the target term has to be balanced (with
respect to the source) first.
\end{enumerate}
When specialized to modularity, the same differences and similarities
can be encountered when comparing \cite{FZM11} to the present work.
Short steps differ in two ways from \cite{vO08b}.
The imbalance is defined differently and the underlying rewrite
sequences are less restricted here. Nevertheless, they define the
same relation on native terms. This covers Theorem~\ref{thm-main-con}.
For Theorems~\ref{thm-main-ll} and~\ref{thm-main-bd},
our proof deals with a new effect, namely fusion from above. This
makes confluence of short short steps ($\RRR{}$) a non-trivial
matter.

We remark that layer systems according to Definition~\ref{def-laysys}
differ from those in~\cite{FZM11}. The latter are closer to
variable-restricted layer systems (Definition~\ref{def-restrict}). Since
the weakened
condition \cvarp is only needed for the order-sorted setting, we decided
to base the theory on the
easier condition \cvar instead and then derive the main results for 
variable-restricted layer systems separately
(cf.\ Section~\ref{sec-restricted-ls}).
Furthermore, we remark that the notions of 
weakly layered and layered
(which are related to weakly consistent and consistent
in~\cite{FZM11}) have changed in an incomparable way, even
for variable-restricted
layer systems.
This is due to the new condition \cstepfusion,
which is required for our constructive proof, as shown
in Example~\ref{ex-stepfusion-essential}.

\subsection{Constructivity}
\label{rel-constructivity}

We say that a TRS is \emph{constructively confluent} if there is a
procedure that, given a peak $t \lud{}{*}\from s \to^* u$, constructs
a valley $t \to^* v \lud{}{*}\from u$. In \cite{vO08b} constructive
confluence is proved to be a modular property for disjoint TRSs.

Most previous proofs of modularity and related results rely on
the reduction of terms until they allow no further fusion,
which requires checking whether the top layer of a term may collapse,
a property which is undecidable.
This includes \cite{T87,KMTV94,O94a,K95,AT96,AT97,JT08,JL12}.
Interestingly, \cite{L96} is constructive, but not applicable
in general as observed at the beginning of this section.

The key observation for obtaining a constructive result is that our
main tool for establishing confluence,
the decreasing diagrams technique, is constructive: If any given
\emph{local} peak can be joined decreasingly in a constructive way,
then any conversion becomes joinable by exhaustively replacing
local peaks by smaller conversions until none are left.

For our proofs to be constructive, the TRS needs to be constructively
confluent on terms of rank one. Furthermore, the proofs rely on the
decomposition of arbitrary terms into their max-top and aliens.
Consequently, we must be able to decide whether a given context
$C \Cleq t$ is a max-top of $t$. 
In the applications from Section~\ref{sec-applications}, this
is indeed the case.

If these two assumptions are satisfied, then our proofs are constructive
and we obtain the following corollary.

\begin{corollary}
\label{cor-constructive}
Let~$\RR$ be a TRS.
Assume that $\RR$ is left-linear and weakly layered,
or bounded duplicating and weakly
layered, or layered. If $\RR$ is constructively confluent on terms of
rank one and for any context $C$ and term $t$ it is
decidable if $C$ is a max-top of $t$, then $\RR$ is
constructively confluent.
\end{corollary}

We remark that the above corollary extends to variable-restricted layer
systems, and thus to the order-sorted application in
Section~\ref{app-order-sorted2}.

\section{Conclusion}
\label{sec-conclusion}

In this article we have presented an abstract layer framework that
covers several known results about the modularity and persistence of
confluence. The framework enabled us to correct the result claimed in
\cite{AT96} on order-sorted persistence, and, by placing
weaker conditions on
left-linear or bounded duplicating systems, to increase its applicability. 
We have incorporated a decomposition technique
based on order-sorted persistence (Theorem~\ref{thm-order}) into
\CSI~\cite{ZFM11b}, our confluence
prover. In the implementation we approximate bounded duplication by
non-duplication.
We also showed how Kahrs' confluence result for curried systems is
obtained as an instance of our layer framework.

As future work, we plan to investigate how to apply layer systems to
other properties of TRSs, like termination or having
unique normal forms. 
Finally, we worked out the technical details of our main results to
prepare for future certification efforts in a theorem prover like
Isabelle or Coq. For the latter it is essential that here
(compared to our previous work~\cite{FZM11}) we based our
setting on the constructive modularity proof in~\cite{vO08b}.
The underlying proof technique, decreasing diagrams, has already been
formalized in Isabelle \cite{Z13,Z13b}.

\begin{acks}
The detailed comments of the anonymous reviewers helped to improve the
presentation.
\end{acks}


\bibliographystyle{acmtrans}
\bibliography{paper}

\begin{received}
Received April 2014; revised October 2014; accepted December 2014
\end{received}
\end{document}

%% file: common.inc
\usepackage{tikz}
\usetikzlibrary{calc}
\usepackage{stmaryrd}

\newcommand{\Pos}{\mathcal{P}\mathsf{os}}
\newcommand{\Fun}{\mathcal{F}\mathsf{un}}
\newcommand{\Var}{\mathcal{V}\mathsf{ar}}

\def\Vtop#1{\setbox0=\hbox{#1}\lower\ht0\hbox{\raise 2ex\box0}}

\makeatletter

\def\lud#1#2#3{%
 \def\next##1##2{%
  \setbox0=\hbox{$##1\vphantom{#3}\@ifempty{#1}{}{_{\vphantom{#1}}}\@ifempty{#2}{}{^{#2}}$}%
  \setbox1=\hbox{$##1\vphantom{#3}\@ifempty{#1}{}{_{#1}}\@ifempty{#2}{}{^{\vphantom{#2}}}$}%
  \setbox2=\vbox{\hbox to\wd0{}\hbox to\wd1{}}%
  \mathrel{\hskip\wd2\hskip-\wd0\box0\hskip-\wd1\box1{#3}}%
 }%
 \mathpalette\next{}%
}

\def\rud#1#2#3{#3\@ifempty{#1}{}{_{#1}}\@ifempty{#2}{}{^{#2}}}

\makeatother

%% file: samples/huet_2_1.tikz
\begin{tikzpicture}[scale=.5]
  \coordinate (r) at (0.0,0.0);
  \coordinate (rl) at (-1.25,-1.0);
  \coordinate (rr) at (1.25,-1.0);
  \coordinate (ra) at (-1.25,-1.0);
  \coordinate (ral) at (-2.25,-2.0);
  \coordinate (rar) at (-0.25,-2.0);
  \coordinate (rb) at (1.25,-1.0);
  \coordinate (rbl) at (0.25,-2.0);
  \coordinate (rbr) at (2.25,-2.0);
  \draw [rounded corners, fill=myblue] (rl) -- (rr) -- (r) -- cycle;
  \draw [rounded corners, fill=myred] (ral) -- (rar) -- (ra) --
        cycle;
  \draw [rounded corners, fill=myred] (rbl) -- (rbr) -- (rb) --
        cycle;
  \node at ($(r)-(0,.6)$) {$\m{f}$};
  \node at ($(ra)-(0,.6)$) {$\m{c}$};
  \node at ($(rb)-(0,.6)$) {$\m{c}$};
\end{tikzpicture}

%% file: samples/huet_3_1.tikz
\begin{tikzpicture}[scale=.5]
  \coordinate (r) at (0.0,0.0);
  \coordinate (rl) at (-1.25,-1.0);
  \coordinate (rr) at (1.25,-1.0);
  \coordinate (ra) at (-1.25,-1.0);
  \coordinate (ral) at (-2.25,-2.0);
  \coordinate (rar) at (-0.25,-2.0);
  \coordinate (rb) at (1.25,-1.0);
  \coordinate (rbl) at (0.25,-2.0);
  \coordinate (rbr) at (2.25,-2.0);
  \coordinate (rba) at (1.25,-2.0);
  \coordinate (rbal) at (0.25,-3.0);
  \coordinate (rbar) at (2.25,-3.0);
  \draw [rounded corners, fill=myblue] (rl) -- (rr) -- (r) -- cycle;
  \draw [rounded corners, fill=myred] (ral) -- (rar) -- (ra) --
        cycle;
  \draw [rounded corners, fill=myred] (rbl) -- (rbr) -- (rb) --
        cycle;
  \draw [rounded corners, fill=myblue] (rbal) -- (rbar) -- (rba) --
        cycle;
  \node at ($(r)-(0,.6)$) {$\m{f}$};
  \node at ($(ra)-(0,.6)$) {$\m{c}$};
  \node at ($(rb)-(0,.6)$) {$\m{g}$};
  \node at ($(rba)-(0,.6)$) {$\m{c}$};
\end{tikzpicture}

%% file: samples/huet_2_2.tikz
\begin{tikzpicture}[scale=.5]
  \coordinate (r) at (0.0,0.0);
  \coordinate (rl) at (-1.25,-1.0);
  \coordinate (rr) at (1.25,-1.0);
  \coordinate (ra) at (-1.25,-1.0);
  \coordinate (ral) at (-2.25,-2.0);
  \coordinate (rar) at (-0.25,-2.0);
  \coordinate (rb) at (1.25,-1.0);
  \coordinate (rbl) at (0.25,-2.0);
  \coordinate (rbr) at (2.25,-2.0);
  \draw [rounded corners, fill=myblue] (rl) -- (rr) -- (r) -- cycle;
  \draw [rounded corners, fill=myred] (ral) -- (rar) -- (ra) --
        cycle;
  \draw [rounded corners, fill=myred] (rbl) -- (rbr) -- (rb) --
        cycle;
  \node at ($(r)-(0,.6)$) {$\m{f}$};
  \node at ($(ra)-(0,.6)$) {$\m{c}$};
  \node at ($(rb)-(0,.6)$) {$\m{c}$};
\end{tikzpicture}

%% file: samples/huet_3_2.tikz
\begin{tikzpicture}[scale=.5]
  \coordinate (r) at (0.0,0.0);
  \coordinate (rl) at (-1.25,-1.0);
  \coordinate (rr) at (1.25,-1.0);
  \coordinate (ra) at (-1.25,-1.0);
  \coordinate (ral) at (-2.25,-2.0);
  \coordinate (rar) at (-0.25,-2.0);
  \coordinate (rb) at (1.25,-1.0);
  \coordinate (rbl) at (0.25,-2.0);
  \coordinate (rbr) at (2.25,-2.0);
  \coordinate (rba) at (1.25,-2.0);
  \coordinate (rbal) at (0.25,-3.0);
  \coordinate (rbar) at (2.25,-3.0);
  \draw [rounded corners, fill=myblue] (rl) -- (rb) -- (rbl) --
        (rbr) -- (rb) -- (r) -- cycle;
  \draw [rounded corners, fill=myred] (ral) -- (rar) -- (ra) --
        cycle;
  \draw [rounded corners, fill=myred] (rbal) -- (rbar) -- (rba) --
        cycle;
  \node at ($(r)-(0,.6)$) {$\m{f}$};
  \node at ($(ra)-(0,.6)$) {$\m{c}$};
  \node at ($(rb)-(0,.6)$) {$\m{g}$};
  \node at ($(rba)-(0,.6)$) {$\m{c}$};
\end{tikzpicture}

%% file: samples/fusion_above_l_2.tikz
\begin{tikzpicture}[scale=.5]
  \coordinate (r) at (0.0,0.0);
  \coordinate (rl) at (-1.25,-1.0);
  \coordinate (rr) at (1.25,-1.0);
  \coordinate (ra) at (-1.25,-1.0);
  \coordinate (ral) at (-2.25,-2.0);
  \coordinate (rar) at (-0.25,-2.0);
  \coordinate (rb) at (1.25,-1.0);
  \coordinate (rbl) at (0.25,-2.0);
  \coordinate (rbr) at (2.25,-2.0);
  \draw [rounded corners, fill=myblue] (ra) -- (ral) -- (rar) --
        (ra) -- (rr) -- (r) -- cycle;
  \draw [rounded corners, fill=myred] (rbl) -- (rbr) -- (rb) --
        cycle;
  \node at ($(r)-(0,.6)$) {$\m{h}$};
  \node at ($(ra)-(0,.6)$) {$\m{c}$};
  \node at ($(rb)-(0,.6)$) {$\m{c}$};
\end{tikzpicture}

%% file: samples/fusion_above_r_2.tikz
\begin{tikzpicture}[scale=.5]
  \coordinate (r) at (0.0,0.0);
  \coordinate (rl) at (-1.0,-1.0);
  \coordinate (rr) at (1.0,-1.0);
  \coordinate (ra) at (0.0,-1.0);
  \coordinate (ral) at (-1.25,-2.0);
  \coordinate (rar) at (1.25,-2.0);
  \coordinate (raa) at (-1.25,-2.0);
  \coordinate (raal) at (-2.25,-3.0);
  \coordinate (raar) at (-0.25,-3.0);
  \coordinate (rab) at (1.25,-2.0);
  \coordinate (rabl) at (0.25,-3.0);
  \coordinate (rabr) at (2.25,-3.0);
  \draw [rounded corners, fill=myblue] (rl) -- (ra) -- (raa) --
        (raal) -- (raar) -- (raa) -- (rar) -- (ra) -- (rr) -- (r) -- cycle;
  \draw [rounded corners, fill=myred] (rabl) -- (rabr) -- (rab) --
        cycle;
  \node at ($(r)-(0,.6)$) {$\m{g}$};
  \node at ($(ra)-(0,.6)$) {$\m{h}$};
  \node at ($(raa)-(0,.6)$) {$\m{c}$};
  \node at ($(rab)-(0,.6)$) {$\m{c}$};
\end{tikzpicture}

%% file: samples/conspiring_l_2.tikz
\begin{tikzpicture}[scale=.5]
  \coordinate (r) at (0.0,0.0);
  \coordinate (rl) at (-1.25,-1.0);
  \coordinate (rr) at (1.25,-1.0);
  \coordinate (ra) at (-1.25,-1.0);
  \coordinate (ral) at (-2.25,-2.0);
  \coordinate (rar) at (-0.25,-2.0);
  \coordinate (raa) at (-1.25,-2.0);
  \coordinate (raal) at (-2.25,-3.0);
  \coordinate (raar) at (-0.25,-3.0);
  \coordinate (rb) at (1.25,-1.0);
  \coordinate (rbl) at (0.25,-2.0);
  \coordinate (rbr) at (2.25,-2.0);
  \draw [rounded corners, fill=myblue] (rl) -- (rr) -- (r) -- cycle;
  \draw [rounded corners, fill=myred] (ral) -- (raa) -- (raal) --
        (raar) -- (raa) -- (rar) -- (ra) -- cycle;
  \draw [rounded corners, fill=myred] (rbl) -- (rbr) -- (rb) --
        cycle;
  \node at ($(r)-(0,.6)$) {$\m{f}$};
  \node at ($(ra)-(0,.6)$) {$\m{g}$};
  \node at ($(raa)-(0,.6)$) {$\m{c}$};
  \node at ($(rb)-(0,.6)$) {$\m{c}$};
\end{tikzpicture}

%% file: samples/conspiring_r_2.tikz
\begin{tikzpicture}[scale=.5]
  \coordinate (r) at (0.0,0.0);
  \coordinate (rl) at (-1.25,-1.0);
  \coordinate (rr) at (1.25,-1.0);
  \coordinate (ra) at (-1.25,-1.0);
  \coordinate (ral) at (-2.25,-2.0);
  \coordinate (rar) at (-0.25,-2.0);
  \coordinate (raa) at (-1.25,-2.0);
  \coordinate (raal) at (-2.25,-3.0);
  \coordinate (raar) at (-0.25,-3.0);
  \coordinate (rb) at (1.25,-1.0);
  \coordinate (rbl) at (0.25,-2.0);
  \coordinate (rbr) at (2.25,-2.0);
  \coordinate (rba) at (1.25,-2.0);
  \coordinate (rbal) at (0.25,-3.0);
  \coordinate (rbar) at (2.25,-3.0);
  \draw [rounded corners, fill=myblue] (ra) -- (ral) -- (raa) --
        (raal) -- (raar) -- (raa) -- (rar) -- (ra) -- (rb) -- (rbl) --
        (rba) -- (rbal) -- (rbar) -- (rba) -- (rbr) -- (rb) -- (r) --
        cycle;
  \node at ($(r)-(0,.6)$) {$\m{f}$};
  \node at ($(ra)-(0,.6)$) {$\m{g}$};
  \node at ($(raa)-(0,.6)$) {$\m{c}$};
  \node at ($(rb)-(0,.6)$) {$\m{g}$};
  \node at ($(rba)-(0,.6)$) {$\m{c}$};
\end{tikzpicture}

%% file: samples/app0.tikz
\begin{tikzpicture}[scale=.5]
  \coordinate (r) at (0.0,0.0);
  \coordinate (rl) at (-3.125,-1.0);
  \coordinate (rr) at (3.125,-1.0);
  \coordinate (ra) at (-3.125,-1.0);
  \coordinate (ral) at (-5.625,-2.0);
  \coordinate (rar) at (-0.625,-2.0);
  \coordinate (raa) at (-5.625,-2.0);
  \coordinate (raal) at (-6.875,-3.0);
  \coordinate (raar) at (-4.375,-3.0);
  \coordinate (raaa) at (-6.875,-3.0);
  \coordinate (raaal) at (-7.875,-4.0);
  \coordinate (raaar) at (-5.875,-4.0);
  \coordinate (raab) at (-4.375,-3.0);
  \coordinate (raabl) at (-5.375,-4.0);
  \coordinate (raabr) at (-3.375,-4.0);
  \coordinate (rab) at (-0.625,-2.0);
  \coordinate (rabl) at (-1.875,-3.0);
  \coordinate (rabr) at (0.625,-3.0);
  \coordinate (raba) at (-1.875,-3.0);
  \coordinate (rabal) at (-2.875,-4.0);
  \coordinate (rabar) at (-0.875,-4.0);
  \coordinate (rabb) at (0.625,-3.0);
  \coordinate (rabbl) at (-0.375,-4.0);
  \coordinate (rabbr) at (1.625,-4.0);
  \coordinate (rb) at (3.125,-1.0);
  \coordinate (rbl) at (2.125,-2.0);
  \coordinate (rbr) at (4.125,-2.0);
  \draw [rounded corners, fill=myblue] (rl) -- (rb) -- (rbl) --
        (rbr) -- (rb) -- (r) -- cycle;
  \draw [rounded corners, fill=myred] (ral) -- (rar) -- (ra) --
        cycle;
  \draw [rounded corners, fill=myblue] (raal) -- (raab) -- (raabl) --
        (raabr) -- (raab) -- (raa) -- cycle;
  \draw [rounded corners, fill=myred] (raaal) -- (raaar) -- (raaa) --
        cycle;
  \draw [rounded corners, fill=myblue] (rabl) -- (rabb) -- (rabbl) --
        (rabbr) -- (rabb) -- (rab) -- cycle;
  \draw [rounded corners, fill=myred] (rabal) -- (rabar) -- (raba) --
        cycle;
  \node at ($(r)-(0,.6)$) {$\m{Ap}$};
  \node at ($(ra)-(0,.6)$) {$\m{Ap}$};
  \node at ($(raa)-(0,.6)$) {$\m{Ap}$};
  \node at ($(raaa)-(0,.6)$) {$\m{f_0}$};
  \node at ($(raab)-(0,.6)$) {$x$};
  \node at ($(rab)-(0,.6)$) {$\m{Ap}$};
  \node at ($(raba)-(0,.6)$) {$\m{f_0}$};
  \node at ($(rabb)-(0,.6)$) {$x$};
  \node at ($(rb)-(0,.6)$) {$x$};
\end{tikzpicture}

%% file: samples/app1.tikz
\begin{tikzpicture}[scale=.5]
  \coordinate (r) at (0.0,0.0);
  \coordinate (rl) at (-3.125,-1.0);
  \coordinate (rr) at (3.125,-1.0);
  \coordinate (ra) at (-3.125,-1.0);
  \coordinate (ral) at (-5.625,-2.0);
  \coordinate (rar) at (-0.625,-2.0);
  \coordinate (raa) at (-5.625,-2.0);
  \coordinate (raal) at (-6.875,-3.0);
  \coordinate (raar) at (-4.375,-3.0);
  \coordinate (raaa) at (-6.875,-3.0);
  \coordinate (raaal) at (-7.875,-4.0);
  \coordinate (raaar) at (-5.875,-4.0);
  \coordinate (raab) at (-4.375,-3.0);
  \coordinate (raabl) at (-5.375,-4.0);
  \coordinate (raabr) at (-3.375,-4.0);
  \coordinate (rab) at (-0.625,-2.0);
  \coordinate (rabl) at (-1.875,-3.0);
  \coordinate (rabr) at (0.625,-3.0);
  \coordinate (raba) at (-1.875,-3.0);
  \coordinate (rabal) at (-2.875,-4.0);
  \coordinate (rabar) at (-0.875,-4.0);
  \coordinate (rabb) at (0.625,-3.0);
  \coordinate (rabbl) at (-0.375,-4.0);
  \coordinate (rabbr) at (1.625,-4.0);
  \coordinate (rb) at (3.125,-1.0);
  \coordinate (rbl) at (2.125,-2.0);
  \coordinate (rbr) at (4.125,-2.0);
  \draw [rounded corners, fill=myblue] (rl) -- (rb) -- (rbl) --
        (rbr) -- (rb) -- (r) -- cycle;
  \draw [rounded corners, fill=myred] (ral) -- (rab) -- (raba) --
        (rabal) -- (rabar) -- (raba) -- (rabb) -- (rabbl) -- (rabbr) --
        (rabb) -- (rab) -- (ra) -- cycle;
  \draw [rounded corners, fill=myblue] (raaa) -- (raaal) --
        (raaar) -- (raaa) -- (raab) -- (raabl) -- (raabr) -- (raab) --
        (raa) -- cycle;
  \node at ($(r)-(0,.6)$) {$\m{Ap}$};
  \node at ($(ra)-(0,.6)$) {$\m{Ap}$};
  \node at ($(raa)-(0,.6)$) {$\m{Ap}$};
  \node at ($(raaa)-(0,.6)$) {$\m{f_1}$};
  \node at ($(raab)-(0,.6)$) {$x$};
  \node at ($(rab)-(0,.6)$) {$\m{Ap}$};
  \node at ($(raba)-(0,.6)$) {$\m{f_1}$};
  \node at ($(rabb)-(0,.6)$) {$x$};
  \node at ($(rb)-(0,.6)$) {$x$};
\end{tikzpicture}

%% file: samples/app2.tikz
\begin{tikzpicture}[scale=.5]
  \coordinate (r) at (0.0,0.0);
  \coordinate (rl) at (-3.125,-1.0);
  \coordinate (rr) at (3.125,-1.0);
  \coordinate (ra) at (-3.125,-1.0);
  \coordinate (ral) at (-5.625,-2.0);
  \coordinate (rar) at (-0.625,-2.0);
  \coordinate (raa) at (-5.625,-2.0);
  \coordinate (raal) at (-6.875,-3.0);
  \coordinate (raar) at (-4.375,-3.0);
  \coordinate (raaa) at (-6.875,-3.0);
  \coordinate (raaal) at (-7.875,-4.0);
  \coordinate (raaar) at (-5.875,-4.0);
  \coordinate (raab) at (-4.375,-3.0);
  \coordinate (raabl) at (-5.375,-4.0);
  \coordinate (raabr) at (-3.375,-4.0);
  \coordinate (rab) at (-0.625,-2.0);
  \coordinate (rabl) at (-1.875,-3.0);
  \coordinate (rabr) at (0.625,-3.0);
  \coordinate (raba) at (-1.875,-3.0);
  \coordinate (rabal) at (-2.875,-4.0);
  \coordinate (rabar) at (-0.875,-4.0);
  \coordinate (rabb) at (0.625,-3.0);
  \coordinate (rabbl) at (-0.375,-4.0);
  \coordinate (rabbr) at (1.625,-4.0);
  \coordinate (rb) at (3.125,-1.0);
  \coordinate (rbl) at (2.125,-2.0);
  \coordinate (rbr) at (4.125,-2.0);
  \draw [rounded corners, fill=myblue] (rl) -- (rb) -- (rbl) --
        (rbr) -- (rb) -- (r) -- cycle;
  \draw [rounded corners, fill=myred] (raa) -- (raaa) -- (raaal) --
        (raaar) -- (raaa) -- (raab) -- (raabl) -- (raabr) -- (raab) --
        (raa) -- (rar) -- (ra) -- cycle;
  \draw [rounded corners, fill=myblue] (rabl) -- (rabb) -- (rabbl) --
        (rabbr) -- (rabb) -- (rab) -- cycle;
  \draw [rounded corners, fill=myred] (rabal) -- (rabar) -- (raba) --
        cycle;
  \node at ($(r)-(0,.6)$) {$\m{Ap}$};
  \node at ($(ra)-(0,.6)$) {$\m{Ap}$};
  \node at ($(raa)-(0,.6)$) {$\m{Ap}$};
  \node at ($(raaa)-(0,.6)$) {$\m{f_2}$};
  \node at ($(raab)-(0,.6)$) {$x$};
  \node at ($(rab)-(0,.6)$) {$\m{Ap}$};
  \node at ($(raba)-(0,.6)$) {$\m{f_2}$};
  \node at ($(rabb)-(0,.6)$) {$x$};
  \node at ($(rb)-(0,.6)$) {$x$};
\end{tikzpicture}

%% file: samples/app3.tikz
\begin{tikzpicture}[scale=.5]
  \coordinate (r) at (0.0,0.0);
  \coordinate (rl) at (-3.125,-1.0);
  \coordinate (rr) at (3.125,-1.0);
  \coordinate (ra) at (-3.125,-1.0);
  \coordinate (ral) at (-5.625,-2.0);
  \coordinate (rar) at (-0.625,-2.0);
  \coordinate (raa) at (-5.625,-2.0);
  \coordinate (raal) at (-6.875,-3.0);
  \coordinate (raar) at (-4.375,-3.0);
  \coordinate (raaa) at (-6.875,-3.0);
  \coordinate (raaal) at (-7.875,-4.0);
  \coordinate (raaar) at (-5.875,-4.0);
  \coordinate (raab) at (-4.375,-3.0);
  \coordinate (raabl) at (-5.375,-4.0);
  \coordinate (raabr) at (-3.375,-4.0);
  \coordinate (rab) at (-0.625,-2.0);
  \coordinate (rabl) at (-1.875,-3.0);
  \coordinate (rabr) at (0.625,-3.0);
  \coordinate (raba) at (-1.875,-3.0);
  \coordinate (rabal) at (-2.875,-4.0);
  \coordinate (rabar) at (-0.875,-4.0);
  \coordinate (rabb) at (0.625,-3.0);
  \coordinate (rabbl) at (-0.375,-4.0);
  \coordinate (rabbr) at (1.625,-4.0);
  \coordinate (rb) at (3.125,-1.0);
  \coordinate (rbl) at (2.125,-2.0);
  \coordinate (rbr) at (4.125,-2.0);
  \draw [rounded corners, fill=myblue] (rl) -- (rb) -- (rbl) --
        (rbr) -- (rb) -- (r) -- cycle;
  \draw [rounded corners, fill=myred] (raa) -- (raaa) -- (raaal) --
        (raaar) -- (raaa) -- (raab) -- (raabl) -- (raabr) -- (raab) --
        (raa) -- (rab) -- (raba) -- (rabal) -- (rabar) -- (raba) --
        (rabb) -- (rabbl) -- (rabbr) -- (rabb) -- (rab) -- (ra) -- cycle;
  \node at ($(r)-(0,.6)$) {$\m{Ap}$};
  \node at ($(ra)-(0,.6)$) {$\m{Ap}$};
  \node at ($(raa)-(0,.6)$) {$\m{Ap}$};
  \node at ($(raaa)-(0,.6)$) {$\m{f_2}$};
  \node at ($(raab)-(0,.6)$) {$x$};
  \node at ($(rab)-(0,.6)$) {$\m{Ap}$};
  \node at ($(raba)-(0,.6)$) {$\m{f_2}$};
  \node at ($(rabb)-(0,.6)$) {$x$};
  \node at ($(rb)-(0,.6)$) {$x$};
\end{tikzpicture}

%% file: samples/qa.tikz
\begin{tikzpicture}[scale=.5]
  \coordinate (r) at (0.0,0.0);
  \coordinate (rl) at (-1.0,-1.0);
  \coordinate (rr) at (1.0,-1.0);
  \draw [rounded corners, fill=myblue] (rl) -- (rr) -- (r) -- cycle;
  \node at ($(r)-(0,.6)$) {$\m{a}$};
\end{tikzpicture}

%% file: samples/qb.tikz
\begin{tikzpicture}[scale=.5]
  \coordinate (r) at (0.0,0.0);
  \coordinate (rl) at (-1.25,-1.0);
  \coordinate (rr) at (1.25,-1.0);
  \coordinate (ra) at (-1.25,-1.0);
  \coordinate (ral) at (-2.25,-2.0);
  \coordinate (rar) at (-0.25,-2.0);
  \coordinate (raa) at (-1.25,-2.0);
  \coordinate (raal) at (-2.25,-3.0);
  \coordinate (raar) at (-0.25,-3.0);
  \coordinate (rb) at (1.25,-1.0);
  \coordinate (rbl) at (0.25,-2.0);
  \coordinate (rbr) at (2.25,-2.0);
  \coordinate (rba) at (1.25,-2.0);
  \coordinate (rbal) at (0.25,-3.0);
  \coordinate (rbar) at (2.25,-3.0);
  \draw [rounded corners, fill=myblue] (ra) -- (ral) -- (rar) --
        (ra) -- (rb) -- (rbl) -- (rbr) -- (rb) -- (r) -- cycle;
  \draw [rounded corners, fill=myred] (raal) -- (raar) -- (raa) --
        cycle;
  \draw [rounded corners, fill=myred] (rbal) -- (rbar) -- (rba) --
        cycle;
  \node at ($(r)-(0,.6)$) {$\m{i}$};
  \node at ($(ra)-(0,.6)$) {$\m{f}$};
  \node at ($(raa)-(0,.6)$) {$\m{c}$};
  \node at ($(rb)-(0,.6)$) {$\m{f}$};
  \node at ($(rba)-(0,.6)$) {$\m{c}$};
\end{tikzpicture}

%% file: samples/qc.tikz
\begin{tikzpicture}[scale=.5]
  \coordinate (r) at (0.0,0.0);
  \coordinate (rl) at (-1.875,-1.0);
  \coordinate (rr) at (1.875,-1.0);
  \coordinate (ra) at (-1.875,-1.0);
  \coordinate (ral) at (-2.875,-2.0);
  \coordinate (rar) at (-0.875,-2.0);
  \coordinate (raa) at (-1.875,-2.0);
  \coordinate (raal) at (-2.875,-3.0);
  \coordinate (raar) at (-0.875,-3.0);
  \coordinate (rb) at (1.875,-1.0);
  \coordinate (rbl) at (0.625,-2.0);
  \coordinate (rbr) at (3.125,-2.0);
  \coordinate (rba) at (0.625,-2.0);
  \coordinate (rbal) at (-0.375,-3.0);
  \coordinate (rbar) at (1.625,-3.0);
  \coordinate (rbaa) at (0.625,-3.0);
  \coordinate (rbaal) at (-0.375,-4.0);
  \coordinate (rbaar) at (1.625,-4.0);
  \coordinate (rbb) at (3.125,-2.0);
  \coordinate (rbbl) at (2.125,-3.0);
  \coordinate (rbbr) at (4.125,-3.0);
  \draw [rounded corners, fill=myblue] (ra) -- (ral) -- (rar) --
        (ra) -- (rb) -- (rba) -- (rbal) -- (rbar) -- (rba) -- (rbr) --
        (rb) -- (r) -- cycle;
  \draw [rounded corners, fill=myred] (raal) -- (raar) -- (raa) --
        cycle;
  \draw [rounded corners, fill=myred] (rbaal) -- (rbaar) -- (rbaa) --
        cycle;
  \draw [rounded corners, fill=myred] (rbbl) -- (rbbr) -- (rbb) --
        cycle;
  \node at ($(r)-(0,.6)$) {$\m{i}$};
  \node at ($(ra)-(0,.6)$) {$\m{f}$};
  \node at ($(raa)-(0,.6)$) {$\m{c}$};
  \node at ($(rb)-(0,.6)$) {$\m{h}$};
  \node at ($(rba)-(0,.6)$) {$\m{e}$};
  \node at ($(rbaa)-(0,.6)$) {$\m{c}$};
  \node at ($(rbb)-(0,.6)$) {$\m{c}$};
\end{tikzpicture}

%% file: samples/qd.tikz
\begin{tikzpicture}[scale=.5]
  \coordinate (r) at (0.0,0.0);
  \coordinate (rl) at (-1.875,-1.0);
  \coordinate (rr) at (1.875,-1.0);
  \coordinate (ra) at (-1.875,-1.0);
  \coordinate (ral) at (-2.875,-2.0);
  \coordinate (rar) at (-0.875,-2.0);
  \coordinate (raa) at (-1.875,-2.0);
  \coordinate (raal) at (-2.875,-3.0);
  \coordinate (raar) at (-0.875,-3.0);
  \coordinate (rb) at (1.875,-1.0);
  \coordinate (rbl) at (0.625,-2.0);
  \coordinate (rbr) at (3.125,-2.0);
  \coordinate (rba) at (0.625,-2.0);
  \coordinate (rbal) at (-0.375,-3.0);
  \coordinate (rbar) at (1.625,-3.0);
  \coordinate (rbb) at (3.125,-2.0);
  \coordinate (rbbl) at (2.125,-3.0);
  \coordinate (rbbr) at (4.125,-3.0);
  \draw [rounded corners, fill=myblue] (ra) -- (ral) -- (rar) --
        (ra) -- (rb) -- (rba) -- (rbal) -- (rbar) -- (rba) -- (rbr) --
        (rb) -- (r) -- cycle;
  \draw [rounded corners, fill=myred] (raal) -- (raar) -- (raa) --
        cycle;
  \draw [rounded corners, fill=myred] (rbbl) -- (rbbr) -- (rbb) --
        cycle;
  \node at ($(r)-(0,.6)$) {$\m{i}$};
  \node at ($(ra)-(0,.6)$) {$\m{f}$};
  \node at ($(raa)-(0,.6)$) {$\m{c}$};
  \node at ($(rb)-(0,.6)$) {$\m{h}$};
  \node at ($(rba)-(0,.6)$) {$\m{c}$};
  \node at ($(rbb)-(0,.6)$) {$\m{c}$};
\end{tikzpicture}

%% file: samples/qe.tikz
\begin{tikzpicture}[scale=.5]
  \coordinate (r) at (0.0,0.0);
  \coordinate (rl) at (-1.25,-1.0);
  \coordinate (rr) at (1.25,-1.0);
  \coordinate (ra) at (-1.25,-1.0);
  \coordinate (ral) at (-2.25,-2.0);
  \coordinate (rar) at (-0.25,-2.0);
  \coordinate (raa) at (-1.25,-2.0);
  \coordinate (raal) at (-2.25,-3.0);
  \coordinate (raar) at (-0.25,-3.0);
  \coordinate (rb) at (1.25,-1.0);
  \coordinate (rbl) at (0.25,-2.0);
  \coordinate (rbr) at (2.25,-2.0);
  \coordinate (rba) at (1.25,-2.0);
  \coordinate (rbal) at (0.25,-3.0);
  \coordinate (rbar) at (2.25,-3.0);
  \coordinate (rbaa) at (1.25,-3.0);
  \coordinate (rbaal) at (0.25,-4.0);
  \coordinate (rbaar) at (2.25,-4.0);
  \draw [rounded corners, fill=myblue] (ra) -- (ral) -- (rar) --
        (ra) -- (rb) -- (rbl) -- (rba) -- (rbal) -- (rbar) -- (rba) --
        (rbr) -- (rb) -- (r) -- cycle;
  \draw [rounded corners, fill=myred] (raal) -- (raar) -- (raa) --
        cycle;
  \draw [rounded corners, fill=myred] (rbaal) -- (rbaar) -- (rbaa) --
        cycle;
  \node at ($(r)-(0,.6)$) {$\m{i}$};
  \node at ($(ra)-(0,.6)$) {$\m{f}$};
  \node at ($(raa)-(0,.6)$) {$\m{c}$};
  \node at ($(rb)-(0,.6)$) {$\m{g}$};
  \node at ($(rba)-(0,.6)$) {$\m{f}$};
  \node at ($(rbaa)-(0,.6)$) {$\m{c}$};
\end{tikzpicture}

%% file: samples/qf.tikz
\begin{tikzpicture}[scale=.5]
  \coordinate (r) at (0.0,0.0);
  \coordinate (rl) at (-1.0,-1.0);
  \coordinate (rr) at (1.0,-1.0);
  \draw [rounded corners, fill=myblue] (rl) -- (rr) -- (r) -- cycle;
  \node at ($(r)-(0,.6)$) {$\m{b}$};
\end{tikzpicture}